%% file: channels-arxiv.tex
\documentclass{amsart}
\usepackage{channelstikz,amsmath,amssymb,enumerate}

\newcommand{\op}{\ensuremath{\textrm{\rm op}}}
\newcommand{\cat}[1]{\ensuremath{\mathbf{#1}}\xspace}
\newcommand{\FHilb}{\cat{FHilb}}

\newcommand{\Rel}{\cat{Rel}}
\newcommand{\Stoch}{\cat{Stoch}}
\newcommand{\CPs}{\ensuremath{\mathbf{CP}^*}\xspace}
\newcommand{\CPM}{\ensuremath{\mathbf{CPM}}\xspace}
\newcommand{\V}{\cat{V}}
\newcommand{\B}{\ensuremath{\mathcal{B}}}
\newcommand{\C}{\ensuremath{\mathbb{C}}}
\newcommand{\M}{\ensuremath{\mathbb{M}}}
\newcommand{\Mat}{\ensuremath{\cat{Mat}}}
\newcommand{\Rpos}{\ensuremath{\mathbb{R}_{\geq 0}}}
\newcommand{\MatR}{\ensuremath{\cat{Mat}(\Rpos)}}
\newcommand{\inprod}[2]{\ensuremath{\langle #1\,|\,#2 \rangle}}
\newcommand{\id}[1][]{\ensuremath{1_{#1}}}
\newcommand{\bra}[1]{\ensuremath{\langle #1 |}}
\newcommand{\ket}[1]{\ensuremath{| #1 \rangle}}
\newcommand{\ie}{\textit{i.e.}\xspace}
\newcommand{\eg}{\textit{e.g.}\xspace}
\DeclareMathOperator{\Tr}{Tr}

\DeclareMathOperator{\Mor}{Mor}
\DeclareMathOperator{\dom}{dom}

\theoremstyle{plain}
\newtheorem{theorem}{Theorem}[section]
\newtheorem{corollary}[theorem]{Corollary}
\newtheorem{lemma}[theorem]{Lemma}
\newtheorem{proposition}[theorem]{Proposition}

\theoremstyle{definition}
\newtheorem{definition}[theorem]{Definition}
\newtheorem{example}[theorem]{Example}
\newtheorem{remark}[theorem]{Remark}

\begin{document}
\title{Categories of Quantum and Classical Channels}
\author{Bob Coecke \and Chris Heunen \and Aleks Kissinger}
\address{Department of Computer Science, University of Oxford}
\email{\{coecke,heunen,alek\}@cs.ox.ac.uk}
\thanks{This research was supported by the Engineering and Physical Sciences Research Council Fellowship EP/L002388/1, and the John Templeton Foundation.}
\date{\today}
\maketitle
\begin{abstract}
  We introduce
  a construction that turns a category of pure state spaces and operators
  into a category of observable algebras and superoperators. 
  For example, it turns the category of finite-dimensional Hilbert
  spaces into the category of finite-dimensional C*-algebras and
  completely positive maps.
  In particular, the new category contains both quantum and classical 
  channels, providing elegant abstract notions of preparation and 
  measurement. We also consider
  nonstandard models, that can be used to investigate which notions
  from algebraic quantum information theory are operationally justifiable.
\end{abstract}

\section{Introduction} 

Algebraic quantum information theory provides a very neat framework in
which to study protocols and algorithms involving both classical and
quantum systems. Instead of stacking structure on top of the base
formalism of Hilbert space -- to accommodate, for example, mixed
states, their measurement and evolution, and classical outcomes --
these basic notions are equal and first-class citizens in the
algebraic approach.

The basic setup is that individual systems are modeled by
C*-algebras, which can be grouped by tensor products, and can evolve
along completely positive maps, also called channels. Classical
systems correspond to commutative algebras. This uniformises many
notions. For example, a density matrix corresponds simply to a channel
from the trivial classical system $\mathbb C$ to a quantum system, and a positive operator valued
measurement is just a channel from a quantum system to a classical
one. One ends up with a category of classical and quantum systems and
channels between them. Advanced protocols can then be modeled by
combining channels in sequence as well as in parallel. For more
information we refer
to~\cite{keyl:quantuminformation,keylwerner:lectures}.  

This paper abstracts that idea away from Hilbert spaces, in an attempt
to obtain a more operational formalism. The Hilbert space formalism is
blessed with such an excess of structure, that many conceptually
different notions coincide in this
model~\cite{redei:instead}. Instead, we will take only the very
basic notion of compositionality as primitive. 

To be precise, we will be working within the programme of categorical
quantum mechanics~\cite{abramskycoecke:cqm,coecke:structures}.  
This programme starts with so-called dagger compact categories,
that assume merely a way of grouping systems together that allows for
entanglement, and a way of composing operations on those systems. A
surprising amount of theory already follows from these primitives,
including scalars, the Born rule, quantum teleportation, and much more.

The categories initially studied mostly accommodated pure
states. However, there is a beautiful construction that works on
arbitrary dagger compact categories, and turns the category containing
pure quantum states and operations
into the category of mixed states and
completely positive maps~\cite{selinger:completelypositive,coeckeheunen:cp}.
The resulting categories can even be axiomatised~\cite{coecke:selinger,coeckeperdrix:channels,coeckeheunen:cp}.
Thus mixed states, and channels between quantum systems, can be studied
without leaving the theory of dagger compact categories. 

Another line of research within categorical quantum mechanics concerns
incorporating classical systems. These can be modeled in terms of the
tensor structure alone, by promoting the no-cloning theorem into an
axiom: the ability to copy and delete becomes an extra feature of
classical systems over quantum ones. This leads to so-called
commutative Frobenius algebras within a
category~\cite{coeckepavlovic:classicalobjects,coeckepavlovicvicary:bases,coeckepaquettepavlovic:structuralism,abramskyheunen:hstar}. Again,
it is pleasantly surprising how much follows: for example, this formalism
encompasses complementary observables, and measurement-based quantum
computing~\cite{coeckeduncan:complementary}. 

This paper combines these two developments in representing quantum channels and
classical systems, respectively. We show that (possibly noncommutative) Frobenius
algebras in the category of Hilbert spaces correspond to
finite-dimensional C*-algebras precisely when they are
normalisable (see also~\cite{vicary:quantumalgebras}). This justifies
regarding such algebras in arbitrary categories as 
\emph{abstract C*-algebras}\footnote{By an \emph{abstract} C*-algebra 
  we mean an object in a monoidal category satisfying certain
  requirements. By a \emph{concrete} one we mean an object satisfying
  those requirements in the category of (finite-dimensional) Hilbert
  spaces. This is not to
  be confused with terminology from functional analysis. There, a
  concrete C*-algebra is a *-subalgebra of the algebra $\B(H)$ of bounded
  operators on a Hilbert space $H$ that is uniformly closed, whereas
  an abstract C*-algebra is any Banach algebra with an involution
  satisfying $\|a^*a\|=\|a\|^2$; these notions are equivalent by the
  Gelfand-Naimark-Segal construction; see \eg~\cite[Theorem~I.9.12]{davidson:cstar}. 
}.
Then, we present a construction that turns a category (of pure states
spaces) into one of channels, in such a way that  the category of
Hilbert spaces becomes the category of finite-dimensional C*-algebras
and channels. We study the cases of ``completely quantum'' and
``completely classical'' abstract C*-algebras, showing that this
so-called CP*--construction neatly combines quantum channels and
classical systems.

Finally, we exemplify our constructions in nonstandard models. This
provides counterexamples that separate conceptually different notions,
even some that are commonly held to coincide. Our results thus form the
starting point for an investigation of the foundations of quantum
mechanics from an operational point of view. For example, one can
show that
commutativity of an algebra of observables need not imply
distributivity of its accompanying quantum logic~\cite{coeckeheunenkissinger:compositional}. The nonstandard
model of sets and relations is a satisfying example of our abstract
theory, which there becomes a theory about the well-studied notion of
a groupoid (see also~\cite{heunencontrerascattaneo:groupoids}). This
opens possibilities to employ ``quantum reasoning'' to obtain group
theoretic results, and vice versa.

Before giving a brief introduction to dagger compact categories, we
end this introduction by reviewing related work. There have been
earlier attempts to combine classical systems with quantum
channels~\cite{selinger:daggeridempotents}. One attempt introduces
biproducts to model classical information. This has the drawback that
classical and quantum information no longer stand on equal footing,
and that adding more primitives than merely compositionality requires
operational justification. Another attempt relies on splitting
idempotents. This is a clean categorical construction that does not
need external ingredients, but it is not so clear that this does not
capture too much. Our CP*--construction mediates between these two
earlier attempts, as made precise in~\cite{heunenkissingerselinger:cpproj}: it needs no external structure,
and it captures the right amount of objects.

A separate development adds classical data to a quantum category via a categorical construction involving
the commutative Frobenius algebras in the category~\cite{coeckepaquettepavlovic:structuralism}. The notion of ``classical morphism''
from that work inspired the formulation of the CP*--construction, by generalising from commutative algebras to non-commutative.

Finally, categorical quantum mechanics links to topological quantum computing~\cite{panangadenpaquette:anyons}. The dagger compact categories of the former form a more basic setting than the modular tensor categories of the latter. Specifically, this article deals with symmetric monoidal categories rather than the more general braided ones. Nevertheless, as the diagrammatic notation of dagger compact categories exemplifies~\cite{selinger:graphicallanguages}, even before going into topological quantum computing, there already is topology in quantum computing. Furthermore, the CP*--construction goes through in a braided setting, for details we refer to the forthcoming book~\cite{heunenvicary:cqm}. This article will avoid those complications and stick to the symmetric setting.

\subsection{Dagger compact categories and graphical language}

It is often useful to reason in a very general sense about processes
and how they compose. Category theory provides the tool to do this. A
category consists of a collection of objects $A, B, C, \ldots$, a
collection morphisms $f, g, \cdots$, an associative operation $\circ$
for (vertical) composition, and for every object $A$ an identity
morphism $\id[A]$. Objects can be thought of as types. They dictate which
morphisms can be composed together. We shall primarily be interested
in categories that have not only a vertical composition operation, but
a horizontal composition as well. 

\begin{definition}\label{def:monoidal-category}
  A \emph{monoidal category} consists of a category $\mathcal V$, an
  object $I \in \mathcal V$ called the monoidal unit, a bifunctor
  $\otimes : \mathcal V \times \mathcal V \rightarrow \mathcal V$
  called the monoidal product, and natural isomorphisms
  $\alpha_{A,B,C} : A \otimes (B \otimes C) \rightarrow (A \otimes B)
  \otimes C$, $\lambda_A : I \otimes A \rightarrow A$, and $\rho_A : A
  \otimes I \rightarrow A$, such that $\lambda_I = \rho_I$ and the
  following diagrams commute: 
  \[
    \begin{tikzpicture}
      \matrix (m) [cdiag,column sep=2em] {
        A \otimes (B \otimes (C \otimes D)) &
        (A \otimes B) \otimes (C \otimes D) &
        ((A \otimes B) \otimes C) \otimes D \\
        A \otimes ((B \otimes C) \otimes D) & &
        (A \otimes (B \otimes C)) \otimes D \\
      };
      \path [arrs]
        (m-1-1) edge node {$\alpha$} (m-1-2)
        (m-1-2) edge node {$\alpha$} (m-1-3)
        (m-1-1) edge node {$A \otimes \alpha$} (m-2-1)
        (m-2-1) edge node {$\alpha$} (m-2-3)
        (m-2-3) edge node {$\alpha \otimes D$} (m-1-3);
    \end{tikzpicture}
  \]
  
  \[
    \begin{tikzpicture}
      \matrix (m) [cdiag] {
        A \otimes (I \otimes B) & & (A \otimes I) \otimes B \\
        & A \otimes B & \\
      };
      \path [arrs]
        (m-1-1) edge node {$\alpha$} (m-1-3)
        (m-1-1) edge node [swap] {$A \otimes \lambda$} (m-2-2)
        (m-1-3) edge node {$\rho \otimes B$} (m-2-2);
    \end{tikzpicture}
  \]
\end{definition}

Our main example is the category $\FHilb$, whose objects are
finite-dimensional complex Hilbert spaces, and whose morphisms are linear
functions. It becomes a monoidal category under the usual tensor
product of Hilbert spaces, with unit object $\mathbb{C}$. 

We often drop $\alpha$, $\lambda$, and $\rho$ when they are clear from
the context. Monoidal categories where all three of these maps are
actually equalities, rather than natural isomorphisms,
are called \emph{strict}.  In any monoidal category, they can be used
to construct a natural isomorphism from some object to any other
bracketing of that object, with or without monoidal units. For
example: 
\[ (A \otimes I) \otimes (B \otimes (I \otimes C)) \cong (A \otimes (B \otimes (C \otimes I))). \]
Mac Lane's \emph{coherence theorem} proves that the equations in
Definition \ref{def:monoidal-category} suffice to show that any such
natural isomorphism is equal to any other
one~\cite{maclane:categories}. This lets us treat monoidal categories
as if they were strict. That is, we may omit brackets, $\alpha$,
$\lambda$, and $\rho$ without ambiguity, simply assuming they are
included where necessary. 

Instead of the usual algebraic notation for morphisms in monoidal
categories, it is often vastly more convenient to use a graphical
notation (see also~\cite{selinger:graphicallanguages}). Morphisms can
be thought of as processes. A morphism takes 
something of type $A$ and produces something of type $B$. We draw
morphisms as:
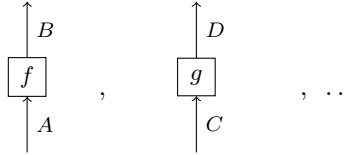
\ctikzfig{cat_boxes}
Identity morphisms are special ``do nothing'' processes, which take
something of type $A$ and return the thing itself. We represent
objects, and the identity morphisms on them, as empty wires:
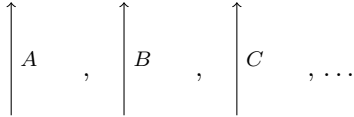
\ctikzfig{cat_wires}
Morphisms are composed by connecting an output wire into an input wire:
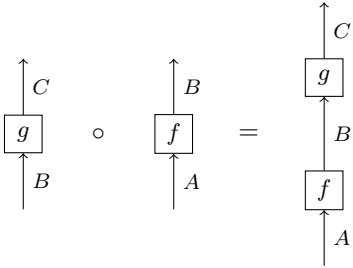
\ctikzfig{cat_compose}
This notation neatly incorporates the assumption that composition
is associative, and that composition with an identity has no effect.

The monoidal product of two 
morphisms is expressed as juxtaposition:
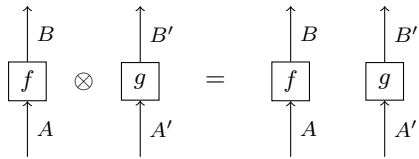
\ctikzfig{cat_tensor_boxes}
The monoidal product is also associative and unital, but possibly only
up to isomorphism. The (identity on) the monoidal unit object $I$ is
denoted by the empty picture.

  
  







\begin{definition}\label{def:braided-monoidal-category}
  A \emph{symmetric monoidal category} is a monoidal category
  with an additional natural isomorphism $\sigma_{A,B} : A \otimes B
  \rightarrow B \otimes A$, such that $\sigma_{A,B}^{-1} =
  \sigma_{B,A}$, $\rho_A = \lambda_A \circ \sigma_{A, I}$, and the
  following ``hexagon'' diagram commutes:
  \[
    \begin{tikzpicture}
      \matrix (m) [cdiag,scale=0.4, column sep=2em] {
        & (B \otimes A) \otimes C & B \otimes (A \otimes C) & \\
        (A \otimes B) \otimes C & & & B \otimes (C \otimes A) \\
        & A \otimes (B \otimes C) & (B \otimes C) \otimes A & \\
      };
      \path [arrs]
        (m-2-1) edge node {$\sigma \otimes C$} (m-1-2)
        (m-1-2) edge node {$\alpha$} (m-1-3)
        (m-1-3) edge node {$B \otimes \sigma$} (m-2-4)
        (m-2-1) edge node {$\alpha$} (m-3-2)
        (m-3-2) edge node {$\sigma$} (m-3-3)
        (m-3-3) edge node {$\alpha$} (m-2-4);
    \end{tikzpicture}
  \]
\end{definition}
  %
We draw symmetry maps as wire crossings:
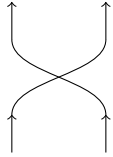
\ctikzfig{cat_sym_crossing}
This graphical notation unambiguously represents morphisms in 
symmetric monoidal
categories~\cite{joyalstreet:braidedtensorcategories}. Moreover, this
representation is sound and complete with respect to the algebraic
definition of a symmetric monoidal category. 
Our example monoidal category $\FHilb$ becomes symmetric by letting
$\sigma_{H,K}(h \otimes k) := k \otimes h$, for all $H$ and $K$.

\begin{definition}\label{def:dual}
  A \emph{compact category} is a symmetric monoidal category in which
  every object $A$ comes with a \emph{dual} object $A^*$ and
  morphisms 
  $\eta_A \colon I \rightarrow A^* \otimes A$ and 
  $\varepsilon_A \colon A \otimes A^* \to I$ satisfying:
  \begin{center}
    \begin{tikzpicture}
      \matrix (m) [cdiag] {
        A & \\
        A \otimes A^* \otimes A & A \\
      };
      \path [arrs]
        (m-1-1) edge node {$1_A$} (m-2-2)
        (m-1-1) edge node [swap] {$A \otimes \eta_A$} (m-2-1)
        (m-2-1) edge node [swap] {$\varepsilon_A \otimes A$} (m-2-2);
    \end{tikzpicture}
    \quad
    \begin{tikzpicture}
      \matrix (m) [cdiag] {
        A^* & A^* \otimes A \otimes A^* \\
            & A^* \\
      };
      \path [arrs]
        (m-1-1) edge node [swap] {$1_A$} (m-2-2)
        (m-1-1) edge node {$\eta_A \otimes A^*$} (m-1-2)
        (m-1-2) edge node {$A^* \otimes \varepsilon$} (m-2-2);
    \end{tikzpicture}
  \end{center}
\end{definition}

In the graphical notation, the object $A^*$ is represented as a wire labelled $A$, but directed downward instead of upward:
\ctikzfig{cat_dual_wire}
We represent $\eta_A$ as a cup, and $\varepsilon_A$ as a cap:
\ctikzfig{cat_cap_cup}
The diagrams from the previous definition are called the ``snake
equations'' because of their graphical representations:
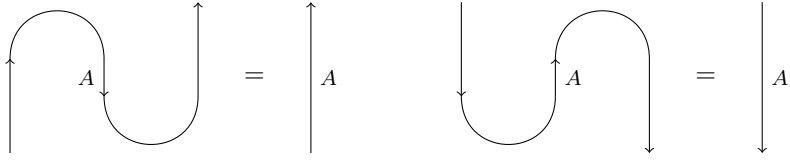
\ctikzfig{cat_line_yank}




In a compact category, any map $f \colon A \to B$ can also be considered as a map $f^* \colon A^* \rightarrow B^*$ by using caps and cups to ``bend the wires'' around:
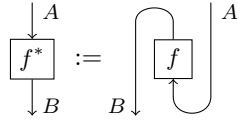
\ctikzfig{cat_bend_wires}

Our example category $\FHilb$ is compact closed. For a
finite-dimensional Hilbert space $H$, let $H^*$ be the dual Hilbert
space. Any orthonormal basis $e_i$ for $H$ then induces a basis
$\overline{e_i}$ for $H^*$. Define $\varepsilon_H(e_i \otimes
\overline{e_j}) = \delta_{ij}$ and $\eta_H(1) = \sum_i \overline{e_i}
\otimes e_j.$ These maps satisfy the snake equations and do not depend
on the choice of basis $e_i$.
Computing $f^* : B^* \to A^*$ in terms of $\varepsilon$ and $\eta$ yields
the (operator) transpose of $f$, \ie~$f^*(\xi) = \xi \circ f$. This is not to be confused
with the matrix transpose, which is basis-dependent (as it depends on fixing
a \textit{particular} isomorphism $A^* \cong A$).

Finally, we abstract the notion of conjugate-transpose.

\begin{definition}
  A \emph{dagger} on a category $\V$ is a contravariant functor
  $\dag \colon \V^\op \to \V$ satisfying $A^\dag = A$ for all objects
  $A$, and $f^{\dag\dag} = f$ for all morphisms $f$. A \emph{unitary} in
  a dagger category is a map $u\colon A \to
  B$ with $u \circ u^\dag = \id[B]$ and $u^\dag \circ u = \id[A]$.
\end{definition}
In particular, $\dagger$-categories are always isomorphic with their opposite category. As the notation suggests, the $\dagger$ functor is an abstract version of the conjugate-transpose of a complex linear map. Thus, for linear maps, the abstract
notion of unitary is precisely the usual one.
\begin{definition}
  A \emph{dagger compact category} is a compact category that comes
  with a dagger such that $(f \otimes g)^\dagger = f^\dagger \otimes g^\dagger$,
  the structure maps $\alpha_A$, $\lambda_A$, and $\sigma_{A,B}$ are all unitary, and $\varepsilon_A^\dagger = \eta_{A^*}$.
\end{definition}
The role of conjugation in a dagger compact category is played by the \textit{lower-star} operation: $f_* : A^* \to B^*$, which is defined as:
\[ f_* := (f^\dagger)^* = (f^*)^\dagger \]


Our example category $\FHilb$ is dagger compact via the formula for
adjoints: $\inprod{f(h)}{k} = \inprod{h}{f^\dag(k)}$.

Finally, we will need the following notion of structure-preserving
functor between dagger compact categories.

\begin{definition}\label{def:monoidal-functor}
  A functor $F \colon \cat{C} \to \cat{D}$ between dagger symmetric monoidal categories
  is a \emph{dagger symmetric monoidal functor} when $F \circ \dagger = \dagger \circ F$ and it comes with an isomorphism $\psi
  \colon I \rightarrow F(I)$ and a natural isomorphism $\varphi_{A,B}
  \colon FA \otimes FB \rightarrow F(A \otimes B)$ making the
  following diagrams commute:  
  \begin{center}
    \begin{tikzpicture}
      \matrix (m) [cdiag] {
        (FA \otimes FB) \otimes FC & F(A \otimes B) \otimes FC & F((A \otimes B) \otimes C) \\
        FA \otimes (FB \otimes FC) & FA \otimes F(B \otimes C) & F(A \otimes (B \otimes C)) \\
      };
      \path [arrs]
        (m-1-1) edge node {$\varphi \otimes FC$} (m-1-2)
        (m-1-2) edge node {$\varphi$} (m-1-3)
        (m-2-1) edge node {$FA \otimes \varphi$} (m-2-2)
        (m-2-2) edge node {$\varphi$} (m-2-3)
        (m-1-1) edge node [swap] {$\alpha$} (m-2-1)
        (m-1-3) edge node {$F(\alpha)$} (m-2-3);
    \end{tikzpicture}
    
    \begin{tikzpicture}
      \matrix (m) [cdiag] {
        FA \otimes I & FA \\
        FA \otimes FI & F(A \otimes I) \\
      };
      \path [arrs]
        (m-1-1) edge node {$\rho$} (m-1-2)
        (m-2-1) edge node {$\varphi$} (m-2-2)
        (m-1-1) edge node [swap] {$FA \otimes \psi$} (m-2-1)
        (m-1-2) edge node {$F(\rho^{-1})$} (m-2-2);
    \end{tikzpicture}
    \qquad
    \begin{tikzpicture}
      \matrix (m) [cdiag] {
        I \otimes FA & FA \\
        FI \otimes FA & F(I \otimes A) \\
      };
      \path [arrs]
        (m-1-1) edge node {$\lambda$} (m-1-2)
        (m-2-1) edge node {$\varphi$} (m-2-2)
        (m-1-1) edge node [swap] {$\psi \otimes FA$} (m-2-1)
        (m-1-2) edge node {$F(\lambda^{-1})$} (m-2-2);
    \end{tikzpicture}

    \begin{tikzpicture}
      \matrix (m) [cdiag] {
        FA \otimes FB & FB \otimes FA \\
        F(A \otimes B) & F(B \otimes A) \\
      };
      \path [arrs]
        (m-1-1) edge node {$\sigma_{FA,FB}$} (m-1-2)
        (m-2-1) edge node {$F\sigma_{A,B}$} (m-2-2)
        (m-1-1) edge node {$\varphi_{A,B}$} (m-2-1)
        (m-1-2) edge node {$\varphi_{B,A}$} (m-2-2);
    \end{tikzpicture}
  \end{center}
\end{definition}

Preserving the dagger and the monoidal structure suffices to preserve the
compact structure~\cite{duncan:thesis}.




\section{Abstract C*-algebras}\label{sec:abstractcstaralgebras}

This section defines so-called normalisable dagger Frobenius
algebras. The running example investigates these structures
in the category of finite-dimensional Hilbert spaces. As will turn
out, they are precisely finite-dimensional C*-algebras. Therefore,
we will think of normalisable dagger Frobenius algebras in arbitrary
dagger compact categories as \emph{abstract C*-algebras}.

\begin{definition}\label{def:frobenius}
  A \emph{dagger Frobenius algebra} is an object $A$ in a dagger
  monoidal category together with morphisms $\whitemult \colon A
  \otimes A \to A$ and $\whiteunit \colon I \to A$, called
  multiplication and unit, satisfying the following diagrammatic
  equations: 
  \ctikzfig{fa_axioms}
  These identities are called associativity, unitality, and the Frobenius law.
  The maps $\whitecomult$ (comultiplication) and 
  and $\whitecounit$ (counit) are defined as $\left(\whitemult\right)^\dagger$ and
  $\left(\whiteunit\right)^\dagger$, respectively.
  They automatically satisfy coassociativity and counitality, which
  are the upside-down versions of associativity and unitality.
\end{definition}

\begin{example}\label{ex:matrixalgebra}
  An important example is the set $A=\M_n$ of $n$-by-$n$ matrices
  with complex entries. This set
  is clearly an algebra: defining $\whitemult$ as $(a,b) \mapsto ab$
  and $\whiteunit \colon \mathbb{C} \to A$ by $1 \mapsto 1_A$
  satisfies associativity and unitality.  The algebra $A$ becomes a
  Hilbert space under the Hilbert--Schmidt inner product
  $\inprod{a}{b} = \Tr(a^\dag b)$. It has a canonical orthonormal
  basis $\{e_{ij} \mid i,j=1,\ldots,n\}$, where $e_{ij}$ is the matrix
  all of whose entries vanish except for a one at location $(i,j)$. We 
  can now compute 
  \[
  \whitecounit(e_{ij}) = \inprod{\whitecounit(e_{ij})}{1} =
  \inprod{e_{ij}}{\whiteunit(1)} = \inprod{e_{ij}}{1_A} = \Tr(e_{ji})
  = \delta_{ij},
  \]
  so that $\whitecounit \colon a \mapsto \Tr(a)$ by linearity.
  Similarly,
  \[
  \inprod{\whitecomult(e_{ij})}{e_{kl} \otimes e_{pq}} =
  \inprod{e_{ij}}{\whitemult(e_{kl} \otimes e_{pq})} =
  \inprod{e_{ij}}{\delta_{lp}e_{kq}} =
  \delta_{ik}\delta_{jq}\delta_{lp},
  \]
  whence $\whitecomult(e_{ij}) = \sum_l e_{il} \otimes e_{lj}$. With
  these explicit expressions it is easy to see
  \begin{align*}
    \whitecomult \circ \whitemult (e_{ij} \otimes e_{kl})
    & = \whitecomult (\delta_{jk} e_{il})
      = \delta_{jk} \sum_p e_{ip} \otimes e_{pl} \\
    & = \sum_p (\whitemult \otimes \whiteline) (e_{ij} \otimes e_{kp} \otimes e_{pl}) \\
    & = (\whitemult \otimes \whiteline) (\sum_p e_{ij} \otimes e_{kp} \otimes e_{pl}) \\
    & = (\whitemult \otimes \whiteline) \circ (\whiteline \otimes \whitecomult)
    (e_{ij} \otimes e_{kl}).
  \end{align*}
  Similarly, 
  \[
      \whitecomult \circ \whitemult (e_{ij} \otimes e_{kl})
      = (\whiteline \otimes \whitemult) \circ (\whitecomult \otimes \whiteline)
    (e_{ij} \otimes e_{kl}).
  \]
  Linearity now shows that $(A,\whitemult,\whitecounit)$ is a dagger
  Frobenius algebra in $\FHilb$. 
\end{example}

Any dagger Frobenius algebra defines a cap and a cup satisfying the snake
identities. 
\begin{equation}\label{eq:frob-snake-ids}
\beginpgfgraphicnamed{cap_cup}
\InputIfFileExists{cap_cup.tikz}{}{\input{./figures/cap_cup.tikz}}
\endpgfgraphicnamed
\end{equation}
This cup and cap provide an alternative form of the Frobenius law that
is sometimes more convenient:
\ctikzfig{frobenius_cup_form}

\begin{definition}\label{def:symmetric}
  A dagger Frobenius algebra $(A,\whitemult,\whiteunit)$ is
  \emph{symmetric} when it satisfies the following equation:
  \ctikzfig{symmetric}
\end{definition}

The dagger Frobenius algebra $\M_n$ in $\FHilb$ is symmetric by the
cyclic property of the trace: $\Tr(ba)=\Tr(ab)$.

\begin{proposition}\label{prop:symmetrictrace}
  For any symmetric Frobenius algebra:
  \ctikzfig{left_right_trace}
\end{proposition}
\begin{proof}
  Symmetry can be used to interchange traces with Frobenius caps and cups.
  \ctikzfig{left_right_trace_pf}
\end{proof}

A dagger Frobenius algebra is certainly symmetric when it is
\emph{commutative}, \ie when it satisfies the following equation:
\ctikzfig{commutative}
Being commutative is strictly stronger than being symmetric. For
example, in $\FHilb$, the algebra $\M_n$ is commutative precisely when
$n=1$. 
Nevertheless, there are plenty of commutative dagger
Frobenius algebras in $\FHilb$. For example, consider the subalgebra
$A$ of $\M_n$ consisting of matrices that are diagonal in some fixed
orthogonal basis. It turns out that this is the only example:
commutative dagger Frobenius algebras $(A,\whitemult,\whiteunit)$ in
$\FHilb$ are in one-to-one correspondence with orthogonal bases of
$A$; see~\cite{coeckepavlovicvicary:bases}. Orthonormal bases
correspond to so-called \emph{normal} algebras.
\footnote{There is a closely related notion called \emph{specialness}. A dagger Frobenius algebra is normal if and only if it is special and symmetric. In $\cat{FHilb}$, normal and special coincide for dagger Frobenius algebras.}
This abstract characterisation of orthonormal bases is
what first sparked the interest in Frobenius algebras in categorical
quantum mechanics~\cite{coeckepavlovic:classicalobjects}.  

We will combine symmetric and commutative algebras as follows.
If $A$ and $B$ are dagger Frobenius algebras in $\FHilb$,
then so is their direct sum $A \oplus B$. If $A$ and $B$ are 
symmetric or commutative, then so is $A \oplus B$. However,
not many interesting, non-commutative algebras in $\FHilb$
are normal, so we need to find a condition to take the place
of normality. Investigating matrix algebras $\mathbb M_n$, 
we might think it suffices to consider algebras that are normal
\textit{up to a scaling factor} $1/n$. However, ``scaled normality'', unlike
normality, is not preserved by direct sum. For example, if
$A=\M_m$ and $B=\M_n$ the induced Frobenius algebra on
of $A \oplus B$ is only normal up to a scalar when $n = m$.

For this reason we will consider a more general condition, called \emph{normalisability}.
Before defining this concept, we introduce the notion of a central map.

\begin{definition}\label{def:central}
  A map $z \colon A \rightarrow A$ is \emph{central} for a
  multiplication $\whitemult$ on $A$ when: 
  \ctikzfig{central}
\end{definition}

The terminology derives from the usual notion of centre for \eg a group,
ring, algebra, etc. Left (or right) multiplication $\whitemult \circ (a \otimes -)
\colon A \to A$ with an element $a \colon I \to A$ is a central map
precisely when $a$ is in the centre $Z(A) = \{ a \in A \mid \forall b
\in A \colon ab = ba \}$. Furthermore, all central maps of a Frobenius
algebra arise this way.
  
A map $g \colon A \to A$ in a dagger category is called
\emph{positive} when $g=h^\dag \circ h$ for some $h$. It is
called \emph{positive definite} if it is a positive isomorphism. Using 
these conditions, we can define normalisability as a
well-behavedness property of the ``loop''. 

\begin{definition}\label{def:normalisable}
  A dagger Frobenius algebra $(A, \whitemult, \whiteunit)$ is
  \emph{normalisable} when it comes with a central, positive definite
  $z \colon A \to A$ such that 
  \ctikzfig{normalisable}
  The map $z$ is called the \textit{normaliser}, and we will often depict it
  simply as $\whitenorm$.
  The algebra is \emph{normal} when we may take $z=1$.
\end{definition}

The equation above uniquely fixes the map $z^2$, so normalisers are unique in any
category where positive square roots are unique, when they exist (such as \FHilb).



All normal Frobenius algebras are symmetric. This turns out the be the
case for dagger normalisable Frobenius algebras as well.

\begin{proposition}\label{prop:normalisableimpliessymmetric}
  Normalisable dagger Frobenius algebras are symmetric.
\end{proposition}
\begin{proof}
  Expand the counit.
  \ctikzfig{norm_symmetric_pf}
  Note that the step marked ($*$) is just a diagram deformation: the two
  multiplication maps have traded places. This corresponds to cyclicity of
  the trace.
\end{proof}

The dagger Frobenius algebra $\M_n$ in $\FHilb$ is normalised
by $z(a)=n^{-1/2} a$:
\ctikzfig{normalisable_pf}
The point of normalisability is that the algebra $\M_m \oplus \M_n$ is
also normalisable (but no longer special unless $m=n$), by the central map $z(a,b)
= (m^{-1/2}a, n^{-1/2}b)$. Thus direct sums $\bigoplus_k \M_{n_k}$ of matrix
algebras are normalisable dagger Frobenius algebras in $\FHilb$.
But these are precisely the finite-dimensional
C*-algebras! This is a standard fact, see
e.g.~\cite[Theorem~III.1.1]{davidson:cstar}. Recall that a
finite-dimensional \emph{C*-algebra} is a finite-dimensional algebra
$A$ equipped with an involution satisfying $\|a^*a\|=\|a\|^2$ (for some
norm satisfying $\|ab\| \leq \|a\|\|b\|$ that is then unique).
The following theorem shows that this exhausts all examples of
normalisable dagger Frobenius algebras 
in $\FHilb$. 
Thus we may think of normalisable dagger Frobenius
algebras in arbitrary categories as \emph{abstract C*-algebras} (see also~\cite{zakrzewski:pseudogroups}).

We can show this directly by defining the C*-algebra structure in terms of the Frobenius algebra
structure. First note that any Frobenius algebra fixes an isomorphism $A^* \cong A$ as follows:
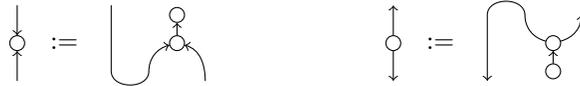
\ctikzfig{dualiser}
These two maps are inverse because of snake identies from equation~\eqref{eq:frob-snake-ids}.

\begin{theorem}\label{thm:algebrasFHilb}
  If $(A,\whitemult,\whiteunit,\whitenorm)$ is a normalisable dagger Frobenius
  algebra in $\FHilb$, then the following involution
  gives it the structure of a finite-dimensional C*-algebra:
  \ctikzfig{star_operation}
  Conversely, up to isomorphism, all finite-dimensional C*-algebras arise in this way.
\end{theorem}
\begin{proof}
  Any dagger Frobenius algebra in $\FHilb$ is a C*-algebra
  under the involution above, and all finite-dimensional
  C*-algebras arise in this way~\cite{vicary:quantumalgebras}, so it
  suffices to prove that any 
  dagger Frobenius algebra $(A, \whitemult, \whiteunit)$ in $\FHilb$
  is normalisable. Since it is unitarily
  isomorphic to a C*-algebra of the form $\bigoplus_k \mathbb{M}_{n_k}$,
  there is an orthonormal basis $\{ e_{ij}^{(k)} : 0 \leq i,j < n_k
  \}$ for $A$, in terms of which $\whitemult$ is defined as
  $e_{ij}^{(k)} \otimes e_{i'j'}^{(k')} \mapsto \delta_{kk'} \delta_{ji'} e_{ij'}^{(k)}$.
  Use this to compute $\Tr_A(\whitemult)$ directly:
  \begin{align*}
    \Tr_A(\whitemult)(e_{ij}^{(k)}) 
    & = \sum_{i'j'k'}
    \left( e_{i'j'}^{(k')} \right)^\dag
    \whitemult \left(
      e_{ij}^{(k)} \otimes
      e_{i'j'}^{(k')}
    \right) \\
    & =
    \sum_{i'j'k'}
    \left( e_{i'j'}^{(k')} \right)^\dag
    \delta_{kk'} \delta_{ji'} e_{ij'}^{(k)} \\
    & = 
    \sum_{j'}
    \left( e_{jj'}^{(k)} \right)^\dag e_{ij'}^{(k)} \\
    & = \sum_{j'} \delta_{ij} 
    = n_k \delta_{ij}.
  \end{align*}
  Also $\whitecounit(e_{ij}^{(k)}) = \delta_{ij}$. Therefore
  $e_{ij}^{(k)} \mapsto n_k^{-1/2} e_{ij}^{(k)}$ defines a normaliser:
  it is positive and invertible, satisfies $\Tr_A(\whitemult)\circ
  (\whitenorm)^2 = \whitecounit$, and acts by a constant scalar on
  each summand of $A$ and so is central.  
\end{proof}

For future reference, we prove two lemmas about abstract C*-algebras, 
including an alternative form of the normalisability condition.
As a matter of convention, we define the following shorthands:
\ctikzfig{action_abbrev}
We can use any such shorthand without ambiguity by stating that we always
preserve the (cylic) ordering of inputs/outputs. That is, the left
input of $\whitemult$ will always be clockwise from the right input, and the right output
 of $\whitecomult$ will always be clockwise from the left output. This rule
 also applies to depictions of $\left( \whitemult \right)^*$ and
 $\left( \whitecomult \right)^*$:
 \ctikzfig{star_abbrev}

\begin{lemma}\label{lem:actions}
  Any symmetric dagger Frobenius algebra satisfies $\;%
\beginpgfgraphicnamed{sfa_ident}
\InputIfFileExists{sfa_ident.tikz}{}{\input{./figures/sfa_ident.tikz}}
\endpgfgraphicnamed\;$.
\end{lemma}
\begin{proof}
  Apply the Frobenius law and associativity.
  \[%
\beginpgfgraphicnamed{sfa_ident_pf}
\InputIfFileExists{sfa_ident_pf.tikz}{}{\input{./figures/sfa_ident_pf.tikz}}
\endpgfgraphicnamed\]
  The middle equation uses symmetry.
\end{proof}

\begin{lemma}\label{lem:normalisability}
  Any normalisable dagger Frobenius algebra satisfies $\;%
\beginpgfgraphicnamed{norm_alt}
\InputIfFileExists{norm_alt.tikz}{}{\input{./figures/norm_alt.tikz}}
\endpgfgraphicnamed\;$.
\end{lemma}
\begin{proof}
  Use centrality of the normaliser, associativity, and unitality.
  \[%
\beginpgfgraphicnamed{norm_alt_pf}
\InputIfFileExists{norm_alt_pf.tikz}{}{\input{./figures/norm_alt_pf.tikz}}
\endpgfgraphicnamed\]
  The marked equation follows from Proposition~\ref{prop:symmetrictrace}.
\end{proof}

To end this section where it started, reconsider the algebra $\M_n$ in
$\FHilb$. It is isomorphic to $(\C^n)^* \otimes \C^n$ by $e_{ij}
\mapsto \bra{i} \otimes \ket{j}$, where $\{\ket{0},\ldots,\ket{n}\}$
is any orthonormal basis of $\C^n$. As it turns out, this way of
constructing C*-algebras works in the abstract, as long as the
category is not too ill-behaved. To be precise, we call an  
object $X$ in a dagger compact category 
\emph{positive-dimensional} if there is a positive definite 
$z \colon I \to I$ satisfying 
\ctikzfig{pos_dim}
A dagger compact closed category is called positive-dimensional if all
its objects are. All the categories we will consider are positive-dimensional.

\begin{proposition}\label{prop:abstractmatrixalgebras}
  In a positive-dimensional dagger compact category, every
  object of the form $H^* \otimes H$ carries a canonical normalisable
  dagger Frobenius algebra with the following multiplication and unit:
  \ctikzfig{pants_alg}
\end{proposition}
\begin{proof}
  It follows immediately from compactness that this is a dagger
  Frobenius algebra.
  Positive-dimensionality provides a positive definite scalar
  $z$ that satisfies $(z^2 \circ \Tr_X(1_X)) \otimes 1_X = 1_X$. Then:
  \ctikzfig{pants_norm}
  Hence $1_{H^*\otimes H} \otimes z$ is a normaliser.
\end{proof}

The abstract C*-algebra of the previous proposition is called an
\emph{abstract matrix algebra}, and is also denoted by $\B(H)$.

\section{Abstract completely positive maps}\label{sec:abstractcpmaps}

Having abstracted C*-algebras from $\FHilb$ to arbitrary categories,
this section does the same for completely positive maps. This will
lead to a fully abstract procedure, called the
\emph{CP*--construction}, that turns any dagger compact category (like \FHilb)
into the category of abstract C*-algebras and abstract completely positive maps.

First recall the definition of completely positive maps between
C*-algebras. An element $a$ of a C*-algebra $A$ is \emph{positive}
when it is of the form $a=b^\star b$ for some $b \in A$. A linear function $f
\colon A \to B$ between C*-algebras is \emph{positive} when it takes
positive elements to positive elements. It is \emph{completely
positive} when the function $f \otimes \id \colon A \otimes
\M_n \to B \otimes \M_n$ is positive for every natural number $n$.
Completely positive maps form a large and well-studied class of
transformations that send (possibly unnormalised) states of open
systems to (possibly unnormalised) states, and hence account for dynamics~\cite{bhatia:positivematrices,paulsen:completelypositive,stormer:positive}.
There is some debate about whether other maps are in fact
unphysical~\cite{alicki:cp,pechukas:cp,shajisudarshan:afraid,zyczkowskibengtsson:cp}.  

This definition translates to abstract C*-algebras as follows: an element $a
\colon I \to A$ of an abstract C*-algebra
$(A,\whitemult,\whiteunit)$  is positive when $a=\whitemult(b^\star
\otimes b)$ for some $b \colon I \to A$. Expanding definitions, we see
that $a \colon I \to A$ is positive when
\ctikzfig{positive_abstract}
for some $b \colon I \to A$.
By Lemma~\ref{lem:normalisability}, this implies:
\ctikzfig{positive_abstract_2}
for some object $X$ and $c \colon I \to X \otimes A$; the middle
equation follows from Lemma~\ref{lem:actions}.

In fact, for Hilbert spaces, the following two characterisations
of positive elements $a$ are equivalent:
\ctikzfig{positive_elem}
However, in other categories, the implication from left to right is
strict. For this reason, we will take the weaker notion to
define an abstract positive element.

This abstract description of positive elements generalises to 
maps $f \colon A \to B$ between abstract C*-algebras
$(A,\whitemult,\whiteunit)$ and $(B,\graymult,\grayunit)$ as follows:
there are an object $X$ and a map $g \colon A \to X \otimes B$ satisfying
\begin{equation}\label{eq:cpstar}
\beginpgfgraphicnamed{cpstar_condition}
\InputIfFileExists{cpstar_condition.tikz}{}{\input{./figures/cpstar_condition.tikz}}
\endpgfgraphicnamed
\end{equation}
The \textit{positive elements} of $A$ are then precisely the maps $I \to A$ 
satisfying this condition. Equation~\eqref{eq:cpstar} is called the
\emph{CP*--condition}. Proposition~\ref{prop:completepositivity} below
shows that this is precisely the right condition to capture complete
positivity abstractly. But before that, the following lemma records that
it indeed makes sense to take tensor products of abstract C*-algebras.

\begin{lemma}\label{lem:tensorproducts}
  If $(A,\whitemult,\whiteunit,\whitenorm)$ and
  $(B,\graymult,\grayunit,\graynorm)$ are normalisable dagger
  Frobenius algebras in a dagger compact category, then so is 
  $(A \otimes B, \prodmult{white dot}{gray dot}, \whiteunit \grayunit,
  \whitenorm \graynorm)$.
\end{lemma}
\begin{proof}
  All the required properties -- associativity, unitality, the
  Frobenius law, and normalisability -- follow easily from the
  graphical calculus for dagger compact categories.
\end{proof}

Incidentally, Lemmas~\ref{lem:actions} and~\ref{lem:normalisability}
provide an alternative form of the CP*--condition that is sometimes
more convenient: equation~\eqref{eq:cpstar} holds if and only if
\begin{equation}\label{eq:cpstarconvolution}
\beginpgfgraphicnamed{cpstar_conv_form}
\InputIfFileExists{cpstar_conv_form.tikz}{}{\input{./figures/cpstar_conv_form.tikz}}
\endpgfgraphicnamed
\end{equation}
for some object $X$ and morphism $h \colon A \to X \otimes B$.

Proposition~\ref{prop:completepositivity} below shows that if a map $A \to
B$ satisfies the CP*--condition~\eqref{eq:cpstar}, then its
composition with another map $I \to A$ satisfying that condition still
satisfies that condition. It is in fact easier to first prove the more general result that the CP*--condition is closed under composition, \ie that maps satisfying~\eqref{eq:cpstar} form a category. In fact, the rest of this section shows that if $\V$ is a dagger
compact category, then so is the category of abstract C*-algebras in
$\V$ and maps satisfying~\eqref{eq:cpstar}, that we now officially define.

\begin{definition}\label{def:cpstar}
  Given a dagger compact category $\V$, we define the data for a new
  category $\CPs[\V]$.
  Objects are normalizable dagger Frobenius algebras in $\V$. 
  Morphisms $(A,\whitemult) \to (B,\graymult)$ are morphisms $f \colon
  A \to B$ in $\V$ satisfying the CP*--condition~\eqref{eq:cpstar}. 
\end{definition}

The next theorem shows that $\CPs[\V]$ is a well-defined category
inheriting composition and identities from $\V$. In fact, it also
inherits tensor products from $\V$ by Lemma~\ref{lem:tensorproducts},
and then becomes a dagger compact category.

\begin{theorem}\label{thm:cpstar}
  If $\V$ is a dagger compact category, $\CPs[\V]$ is again a
  well-defined dagger compact category.
\end{theorem}
\begin{proof}
  Identity maps $\id[A] \colon (A, \whitemult) \to (A,
  \whitemult)$ satisfy the $\CPs$-condition by Lemma~\ref{lem:actions},
  where the role of $g$ in equation~\eqref{eq:cpstar} is played by
  $\whitecomult$.
  
  Next, suppose $f \colon (A, \whitemult) \to (B, \graymult)$ and
  $g \colon (B, \graymult) \to (C, \blackmult)$ satisfy the
  CP*--condition. It then follows from Lemma~\ref{lem:normalisability}
  that their composition does, too.
  \ctikzfig{compose_cpstar}
  Thus $\CPs[\V]$ is indeed a well-defined category.
  
  Lemma~\ref{lem:tensorproducts} gives monoidal structure on the level
  of objects. Given a morphism $f \colon (A, \whitemult) \rightarrow (C,
  \dotmult{alt white dot})$ with Kraus map $h$, and $g \colon (B, \graymult) \rightarrow (D,
  \dotmult{alt gray dot})$ with Kraus map $i$, then 
  $f \otimes g \colon
       (A \otimes B, \prodmult{white dot}{gray dot}) \rightarrow
       (C \otimes D, \prodmult{alt white dot}{alt gray dot})$
  satisfies the CP*--condition: 
  \ctikzfig{cpstar_monoidal}
  Note that $(I, \rho_I)$, where $\rho_I \colon I \otimes I \to
  I$ is the coherence isomorphism of $\V$,  is a normalisable dagger
  Frobenius algebra by the coherence theorem.
  Using this definition of $\otimes$ and $I$, the coherence isomorphisms
  $\alpha$, $\lambda$, and $\rho$ from $\V$ trivially satisfy the
  CP*--condition. Thus $\CPs[\V]$ is a monoidal category. 
  
  To show that $\CPs[\V]$ inherits symmetry, it suffices to
  show that the swap map $\sigma_{A,B} \colon A \otimes B \rightarrow B
  \otimes A$ of $\V$ lifts to a morphism 
  $\sigma_{A,B} \colon
      (A \otimes B, \prodmult{white dot}{gray dot}) \to
      (B \otimes A, \prodmult{gray dot}{white dot})$
  in $\CPs[\V]$. This can be done with two applications of Lemma \ref{lem:actions}.
  \ctikzfig{cpstar_symm}
  Thus $\CPs[\V]$ is a symmetric monoidal category.
  
  The category $\CPs[\V]$ also inherits the dagger from $\V$. If $f
  \colon (A, \whitemult) \rightarrow (B, \graymult)$ satisfies~\eqref{eq:cpstar},
   then so too does $f^\dag$: because
  $
    \whiteaction \circ f^\dag \circ \graycoaction =
    \left( \grayaction \circ f \circ \whitecoaction \right)^\dag
  $, 
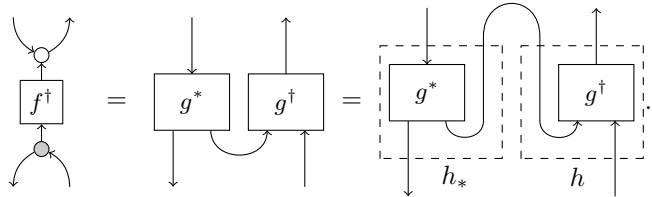
  \[%
\beginpgfgraphicnamed{cpstar_dagger}
\InputIfFileExists{cpstar_dagger.tikz}{}{\input{./figures/cpstar_dagger.tikz}}
\endpgfgraphicnamed.\]
  Since the coherence isomorphisms of $\CPs[\V]$ are those of $\V$,
  they are unitary, and thus $\CPs[\V]$ is a dagger symmetric monoidal category.

  Finally, for compactness, let $(A, \whitemult)$ be an object in
  $\CPs[\V]$. Let $A^*$ be a dual of $A$, with cap
  $\varepsilon_{A^*} \colon A^* \otimes A \to I$.
  If a Frobenius algebra is dagger normalisable, so too are the opposite algebra
  and the transposed algebra (\ie the dual).
  Thus $(A^*,\whitedualmult)$ is a well-defined
  object of $\CPs[\V]$. Now, $ \varepsilon_{A^*} \colon (A, \whitemult) \otimes
  (A^*, \whitedualmult) \rightarrow I $ 
  satisfies the CP*--condition, again by Lemma~\ref{lem:actions}:
  \ctikzfig{cpstar_cap}
  We have already showed that the dagger of a map satisfying the CP*--condition
  also satisfies the CP*--condition, so finally let $\eta_{A} = \varepsilon_{A^*}^\dagger$.
  We complete the proof by noting that $\sigma$, $\eta$, and $\varepsilon$ are
  all defined with the same underlying maps as in \V, so the symmetry and snake equations
  are automatically satisfied.
\end{proof}

We have constructed a category whose objects are abstract C*-algebras,
and we claim that the morphisms are abstract completely positive
maps. The following proposition justifies that claim, by showing
that maps satisfying the CP*--condition correspond exactly to
maps that are completely positive
in the usual sense, in that $f : A \to B$ applied to a positive
(open) state preserves positivity.

\begin{proposition}\label{prop:completepositivity}
  Let $(A,\whitemult)$ and $(B, \graymult)$ be normalisable dagger
  Frobenius algebras and $f \colon A \to B$ a morphism in a dagger
  compact category. The following are equivalent:
  \begin{enumerate}[(a)]
  \item $f$ satisfies the CP*--condition~\eqref{eq:cpstar};
  \item $f \otimes \id[C]$ sends positive elements
    of $(A,\whitemult) \otimes (C, \blackmult)$ to positive elements
    of $(B, \graymult) \otimes (C, \blackmult)$ for all normalisable
    dagger Frobenius algebras $(C, \blackmult)$;
  \item $f \otimes \id[X^* \otimes X]$ sends positive elements
    of $(A,\whitemult) \otimes (X^* \otimes X,\pantsalg)$ to positive
    elements of $(B, \graymult) \otimes (X^* \otimes X,\pantsalg)$ for 
    all objects $X$.
  \end{enumerate}
\end{proposition}
\begin{proof}
  For (a) $\Rightarrow$ (b): if $\rho$ is a positive element of
  $(A,\whitemult) \otimes (C,\blackmult)$, then it can be regarded as
  a morphism $\rho \colon I \to (A,\whitemult) \otimes (C,\blackmult)$
  in $\CPs[\V]$. It then follows from
  Theorem~\ref{thm:cpstar} that $(f \otimes \id[C]) \circ \rho$ is
  also a morphism in $\CPs[\V]$, so it must also be a positive element.
  The implication (b) $\Rightarrow$ (c) is trivial. 
  Finally, for (c) $\Rightarrow$ (a): setting
  $X=A^*$, the following is a positive element of
  $(B,\blackmult) := (A,\whitemult) \otimes (X^* \otimes X,
  \pantsalg)$. 
  \ctikzfig{cpstar_alt_matrix_point}
  Indeed, graphical rewriting using the Frobenius law and symmetry
  shows:
  \ctikzfig{cpstar_alt_matrix_point_ii}
  So, by assumption, $(f \otimes \id[A^*]) \circ \rho$  is also a positive
  element. Applying white caps to both sides establishes that $f$
  satisfies the CP*--condition.
  \[ %
\beginpgfgraphicnamed{cpstar_alt_pf}
\InputIfFileExists{cpstar_alt_pf.tikz}{}{\input{./figures/cpstar_alt_pf.tikz}}
\endpgfgraphicnamed \]
  This finishes the proof.
\end{proof}

The previous proposition is a fully abstract version of
Stinespring's dilation theorem~\cite{stinespring}, or rather (because
our abstract C*-algebras are finite-dimensional) of Choi's
theorem~\cite{choi}. The morphism $g$ in equation~\eqref{eq:cpstar}
therefore called a \emph{Kraus map} for $f$; we emphasise that it is
not unique. Traditional formulations in $\FHilb$ allow a sum of Kraus
maps; this is expressed abstractly by the indexing object $X$
in~\eqref{eq:cpstar}. 

The abstract C*-algebras $(C,\blackmult)$ and $(X^* \otimes X,
\pantsalg)$ in the previous proposition are called the ancillary
system, or \emph{ancilla}. 
In these terms, the previous proposition shows that the
CP*--condition~\eqref{eq:cpstar} characterises those maps that
preserve positivity even when their input and output systems are
regarded as open subsystems of larger systems.
In fact, the previous proposition does slightly better
than Choi's theorem, because the ancilla can be an arbitrary
abstract C*-algebra instead of just an abstract matrix algebra. 

Because of the way we have modeled the definition of $\CPs[\V]$ after
the case of $\FHilb$, the category $\CPs[\FHilb]$ is indeed that of
(concrete) finite-dimensional C*-algebras and completely positive maps, as the
following proposition records.

\begin{proposition}\label{prop:cpstarfhilb}
  $\CPs[\FHilb]$ is equivalent to the category of finite-dimensional
  C*-algebras and completely positive maps. 
\end{proposition}
\begin{proof}
  Define a functor $E$ from $\CPs[\FHilb]$ to the category of
  finite-dimensional C*-algebras and completely positive maps, acting
  on objects as in Theorem~\ref{thm:algebrasFHilb} and as the identity on
  morphisms. This functor is then essentially surjective on objects by
  that theorem.
  Furthermore, Proposition~\ref{prop:completepositivity} shows that
  $E(f)$ is a completely positive map between concrete C*-algebras if
  and only if $f$ satisfies the CP*--condition. This makes $E$ a
  well-defined functor that is full.
  It is faithful by construction, and hence it is an equivalence of
  categories.
\end{proof}

\begin{remark}\label{rem:real}
  We have employed complex Hilbert spaces. It is natural to
  wonder about performing the CP*--construction on real
  finite-dimensional Hilbert spaces.  

  On the level of objects, Theorem~\ref{thm:algebrasFHilb} still goes
  through: abstract C*-algebras in the category of real
  finite-dimensional Hilbert spaces correspond to so-called
  finite-dimensional \emph{real C*-algebras}
  (see~\cite{li:realoperatoralgebras}). However, these need not be
  direct sums of complex matrix algebras; rather, they are 
  direct sums of algebras of matrices over the real numbers, complex
  numbers, or over the
  quaternions~\cite[Theorem~5.7.1]{li:realoperatoralgebras}.  

  On the level of morphisms, Proposition~\ref{prop:completepositivity}
  still holds. However, in the real case these morphisms do not give all
  completely positive
  maps~\cite[Theorem~4.3]{ruan:realoperatorspaces}. The underlying
  issue is that there are more positive elements in real C*-algebras
  than those of the form $a^*a$.
\end{remark}

In the concrete case, *-homomorphisms between C*-algebras are
automatically completely positive. We conclude this section by proving
this holds fully abstractly, providing an easy way to show that some
maps are morphisms in $\CPs[\V]$. 

\begin{definition}\label{def:starhomomorphism}
  If $(A, \whitemult)$ and $(B, \graymult)$ are dagger normalisable
  Frobenius algebras, a morphism $f \colon A \to B$ is
  called a \emph{*-homomorphism} when it satisfies the following equations.
  \ctikzfig{starhomomorphism}
\end{definition}

\begin{lemma}\label{lem:starhomomorphism}
  Let $(A,\whitemult)$ and $(B,\graymult)$ be dagger normalisable
  Frobenius algebras in a dagger compact category $\V$. If $f \colon A
  \to B$ is a *-homomorphism, then it is a well-defined morphism in $\CPs[\V]$.
\end{lemma}
\begin{proof}
  Graphical manipulation shows the following.
  \ctikzfig{starhomoiscp}
  Hence this morphism is a composition of $f^* \otimes f$ and
  $\grayaction \circ \graycoaction$. 
  Both morphisms are completely positive, \ie of the form of the
  right-hand side of equation~\eqref{eq:cpstar}: the former by
  construction, the latter by Lemma~\ref{lem:actions}.
  Therefore $f$ is also completely positive by
  Theorem~\ref{thm:cpstar}, and hence a morphism in $\CPs[\V]$.
\end{proof}

\section{Completely classical systems and completely quantum systems}\label{sec:classicalquantum}

As discussed in Section~\ref{sec:abstractcstaralgebras}, commutative
abstract C*-algebras $(A,\whitemult)$ in $\FHilb$ correspond to
orthogonal bases of $A$. More precisely, the basis vectors are the
\emph{copyable points}, \ie morphisms $p \colon I \to A$ that
satisfy $\whitecomult \circ p = p \otimes p$. 
Expanding arbitrary vectors in this basis, one can show
that the normalised positive elements of $A$ are precisely those 
vectors with positive coefficients summing to $1$. 
Thus, normalised positive elements of a commutative abstract
C*-algebra may be regarded as probability distributions over its
copyable points.  
That is, we may think of commutative abstract C*-algebras as
``completely classical'' systems.\footnote{Commutativity might be
too strong a notion of ``completely classical'' system in the
abstract. A weaker notion of broadcastability, that coincides with
commutativity in $\FHilb$, seems more reasonable. Subsequent work will
investigate such more operational notions of classicality.} 

On the other hand, Section~\ref{sec:abstractcstaralgebras} showed that
abstract matrix algebras can be regarded as ``completely quantum''
systems: their states have no probabilistical mixing aspect at
all. 
In general, abstract C*-algebras are combinations of ``completely
classical'' and ``completely quantum'' parts. 
This section focuses on these two extreme cases. It proves that the
CP*--construction subsumes earlier constructions that remained
separate: 
the Stoch--construction into its ``completely classical''
part~\cite{coeckepaquettepavlovic:structuralism},
and the so-called CPM--construction into its ``completely
quantum''
part~\cite{selinger:completelypositive,coeckeheunen:cp,boixoheunen:sequential}. 
Thus the CP*--construction combines the two, and places classical and
quantum systems and channels on an equal footing in a single category.

\subsection{Completely classical systems}\label{subsec:classical}

First, recall the
Stoch--construction~\cite{coeckepaquettepavlovic:structuralism}.
Like the CP*--construction of the previous section, it turns a dagger
compact category $\V$ into a new one, $\Stoch[\V]$. It will turn out
that it is precisely the full subcategory of $\CPs[\V]$ consisting of
commutative abstract C*-algebras, and that we may regard it as the
subcategory of classical channels.

Objects of $\Stoch[\V]$ are commutative normalisable dagger Frobenius
algebras. Morphisms $(A,\whitemult) \to (B,\blackmult)$ in
$\Stoch[\V]$ are morphisms $f \colon A \to B$ in $\V$ with
\begin{equation}\label{eq:stoch}
\beginpgfgraphicnamed{stoch_condition}
\InputIfFileExists{stoch_condition.tikz}{}{\input{./figures/stoch_condition.tikz}}
\endpgfgraphicnamed
\end{equation}
for some commutative normalisable dagger Frobenius algebra
$(X,\graymult)$ and a morphism $g \colon A \to X \otimes B$ in
$\V$. Here, the conjugation $g_*$ is taken with respect to the caps
and cups induced by $\whitecap$, $\graycap$, and $\blackcap$.

\begin{theorem}\label{thm:stoch}
  For a dagger compact category $\V$, the category $\Stoch[\V]$ is
  isomorphic to the full subcategory of $\CPs[\V]$ consisting of all
  commutative normalisable dagger Frobenius algebras. 
\end{theorem}
\begin{proof}
  We show that~\eqref{eq:stoch} implies~\eqref{eq:cpstar}.
  \ctikzfig{stoch_proof}
  The converse holds since the dualisers $\whiteidualiser$,
  $\grayidualiser$, and $\blackidualiser$, are always invertible.
\end{proof}

The following corollary justifies thinking of $\Stoch[\V]$ as a
category of classical channels.
We call a morphism $f \colon (A,\whitemult)
\to (B,\graymult)$ in $\CPs[\V]$ \emph{normalised} if it preserves
counits: $\graycounit \circ f = \whitecounit$.  
Recall that a \emph{stochastic map} between finite-dimensional Hilbert
spaces is a matrix with positive real entries whose every column sums
to one. 

\begin{corollary}\label{corollary:stochfhilb}
  Normalised morphisms in $\Stoch[\FHilb]$ correspond to stochastic
  maps between finite-dimensional Hilbert spaces.
\end{corollary}
\begin{proof}
  Combine Theorem~\ref{thm:stoch}, Proposition~\ref{prop:cpstarfhilb}
  and~\cite[3.2.3 and 2.1.3]{keyl:quantuminformation}. 
\end{proof}


\subsection{Completely quantum systems}\label{subsec:quantum}

First, we recall the
CPM--construction~\cite{selinger:completelypositive,coeckeheunen:cp}.   
Like the CP*--construction of the previous section, it turns a dagger
compact category $\V$ into a new one, $\CPM[\V]$. It will turn out
that it is precisely the full subcategory of $\CPs[\V]$ consisting of
abstract matrix algebras $\B(H)=(H^* \otimes H,\pantsalg)$, that are
simply identified with $H$, and that we may regard it as the
subcategory of quantum channels. 

Objects of $\CPM[\V]$ are the same as those of $\V$, and morphisms $f
\colon A \to B$ in $\CPM[\V]$ are morphisms $f \colon A^* \otimes A \to
B^* \otimes B$ in $\V$ for which there exist an object $X$ and a
morphism $g \colon A \to X \otimes B$ satisfying:
\ctikzfig{cp_condition}
Composition, identity maps, and $\otimes$ on objects of $\CPM[\V]$ are
as in $\V$. The tensor product is defined on morphisms of $\CPM[\V]$ as
follows: 
\ctikzfig{cp_monoidal}
$\CPM[\V]$ inherits symmetry and compact structure from $\V$,
only ``doubled''. 
\ctikzfig{cp_symm_compact}

The following theorem proves that $\CPM[\V]$ embeds in $\CPs[\V]$,
preserving all structure. To formulate that embedding, recall that a
functor $F$ is dagger symmetric monoidal if it comes with 
natural unitary isomorphisms $\varphi_{A,B} \colon F(A \otimes B)
\rightarrow F(A) \otimes F(B)$ satisfying $\varphi_{I,A} =
\varphi_{A,I} = \id[A]$ and
\[
  \begin{tikzpicture}[xscale=4,yscale=1.5]
    \node (dl) at (0,0) {$F(A \otimes B) \otimes F(C)$};
    \node (dr) at (1,0) {$F(A) \otimes F(B) \otimes F(C)$};
    \node (ul) at (0,1) {$F(A \otimes B \otimes C)$};
    \node (ur) at (1,1) {$F(A) \otimes F(B \otimes C)$};
    \draw[->] (ul) to node[auto,swap] {$\varphi_{A \otimes B, C}$} (dl);
    \draw[->] (ul) to node[auto] {$\varphi_{A, B \otimes B}$} (ur);
    \draw[->] (dl) to node[auto,swap] {$\varphi_{A, B} \otimes \id[F(C)]$} (dr);
    \draw[->] (ur) to node[auto,swap] {$\id[F(A)] \otimes \varphi_{B,C}$} (dr);
  \end{tikzpicture}
  \hspace*{-3mm}
  \begin{tikzpicture}[xscale=3,yscale=1.5]
    \node (dl) at (0,0) {$F(B \otimes A)$};
    \node (dr) at (1,0) {$F(B) \otimes F(A)$};
    \node (ul) at (0,1) {$F(A \otimes B)$};
    \node (ur) at (1,1) {$F(A) \otimes F(B)$};
    \draw[->] (ul) to node[auto,swap] {$F(\sigma_{A,B})$} (dl);
    \draw[->] (ul) to node[auto] {$\varphi_{A,B}$} (ur);
    \draw[->] (dl) to node[auto,swap] {$\sigma_{B,A}$} (dr);
    \draw[->] (ur) to node[auto,swap] {$\sigma_{F(A),F(B)}$} (dr);
  \end{tikzpicture}
\]
For simplicity, we have assumed that the categories involved are
strict monoidal. 

\begin{theorem}\label{thm:cpm}
  If $\V$ is a positive-dimensional dagger compact category,
  \begin{align*}
    \B(A) = (A^* \otimes A, \pantsalg)
    \qquad
    \B(f) = f
  \end{align*}
  defines a functor $\B \colon \CPM[\V] \to \CPs[\V]$ that is full,
  faithful, and dagger symmetric monoidal. 
\end{theorem}
\begin{proof}
  First of all, $\B$ is well-defined, because a morphism $f\colon A^*
  \otimes A \rightarrow B^* \otimes B$ in $\V$ determines a morphism
  $A \to B$ in $\CPM[\V]$ precisely when if it determines a morphism
  $(A^* \otimes A, \pantsalg) \to (B^* \otimes B, \pantsalg)$ in
  $\CPs[\V]$. Indeed, if $f$ is a morphism in $\CPM[\V]$, it also
  satisfies the CP*--condition:
  \ctikzfig{cp_implies_cpstar}
  Conversely, if $f$ is in $\CPs[\V]$, then
  it is also in $\CPM[\V]$:
  \ctikzfig{cpstar_implies_cp}
  
  Composition is defined identically in $\CPM[\V]$ and $\CPs[\V]$, so
  $\B$ is functorial, full, and faithful. 
  
  Define $\varphi_{A,B} \colon \B(A \otimes B) \to
  \B(A) \otimes \B(B)$ as the following ``reshuffling map''.
  \ctikzfig{reshuffle_map}
  To verify that this defines a morphism in $\CPs[\V]$, it suffices to
  show that it is a *-homomorphism by Lemma~\ref{lem:starhomomorphism}. 
  \ctikzfig{reshuffle_hm_mult}
  \ctikzfig{reshuffle_hm_star}
  Next, we show naturality of $\varphi$:
  \ctikzfig{reshuffle_natural}
  The last thing that remains to be shown is coherence for $\varphi$
  with respect to the symmetric monoidal structure. For associativity:
  \ctikzfig{reshuffle_assoc}
  As for the unit equations:
  \ctikzfig{reshuffle_unit}
  Finally, as for symmetry:
  \ctikzfig{reshuffle_symm}
  Thus $\B$ is a full, faithful, dagger symmetric monoidal functor.
\end{proof}

As a consequence of the previous theorem and
Proposition~\ref{prop:cpstarfhilb}, the category $\CPM[\FHilb]$ is
equivalent to the category of matrix algebras and completely positive
maps. This justifies thinking of the ``completely quantum'' part of
$\CPs[\V]$ as a category of quantum channels.

The category $\CPM[\FHilb]$ is strictly smaller than the category
$\CPs[\FHilb]$ of all finite-dimensional C*-algebras and completely
positive maps. That is, the embedding $\B$ of the previous theorem
does not extend to an equivalence of categories: for example, the
finite-dimensional C*-algebra $A=\mathbb{M}_1 \oplus  \mathbb{M}_2$
cannot be isomorphic to a matrix algebra $\M_n$ because
$\dim(A)=1^2+2^2=5\neq n^2=\dim(\mathbb{M}_n)$.  

In analogy to the case $\V=\FHilb$, it stands to reason to regard
objects $H$ of $\V$ as systems whose state space consists of pure
states, and objects $\B(H)$ of $\CPs[\V]$ as systems whose state space
consists of mixed states. So one might think that the ``pure'' category
$\V$ should embed into the ``mixed'' category $\CPs[\V]$. The following
corollary shows that this is indeed the case. 

\begin{corollary}\label{cor:pure}
  If $\V$ is a dagger compact category, 
  \[
    A \mapsto \B(A)
    \qquad
    f \mapsto f_* \otimes f
  \]
  defines a dagger symmetric monoidal functor $\V \to \CPs[\V]$.
\end{corollary}
\begin{proof}
  Combine the previous theorem
  with~\cite[Theorem~4.20]{selinger:completelypositive}. 
\end{proof}

There are no meaningful functors in the opposite directions.
A construction $\CPs[\V] \to \V$ would model decoherence, which
cannot be a structure preserving functor. More precisely, the functor
$\V \to \CPs[\V]$ does not have any adjoints, because it does not
preserve (co)limits: $\B(H \oplus K) \not\cong \B(H) \oplus \B(K)$ for
nontrivial Hilbert spaces $H$ and $K$.  
Similarly, a functor $\CPs[\V] \to \CPM[\V]$ would need to coherently
turn an (abstract) C*-algebra into an (abstract) matrix
algebra. Again, it cannot be an adjoint because it cannot preserve
(co)limits.

\section{Nonstandard models}\label{sec:nonstandard}

So far, we have abstracted classical and quantum systems and channels
from the category $\FHilb$ to arbitrary dagger compact categories
$\V$. Now it is high time to see some other examples. This section
considers three: the category of sets and relations, the category of
matrices with positive entries, and the category of relations with
values in a cancellative quantale. We will see that abstract
C*-algebras in these categories turn out to be important well-known
structures, that are nevertheless quite different from concrete
C*-algebras. 

\subsection{Relations}

First, recall the category $\Rel$. Its objects are sets, and morphisms
$A \to B$ are relations $R \subseteq A \times B$. The composition of
$R \colon A \to B$ and $S \colon B \to C$ is given by
\[
  S \circ R = \{ (a,c) \in A \times C \mid \exists b \in B \colon
  (a,b) \in R, (b,c) \in S \},
\]
and $\{(a,a) \mid a \in A\}$ is the identity on $A$. Cartesian product
makes $\Rel$ into a compact category. Finally, it becomes a dagger
compact category by
\[
  R^\dag = \{(b,a) \mid (a,b) \in R\}.
\]

We start by investigating the objects of $\CPs[\Rel]$. This
immediately shows that these nonstandard abstract C*-algebras are
quite different from C*-algebras (in $\FHilb$): they are precisely
groupoids. Recall that a \emph{groupoid} is a category whose
morphisms are all invertible~\cite{maclane:categories}.

\begin{proposition}\label{prop:cpstarrelobjs}
  Normalisable dagger Frobenius algebras in $\Rel$ are (in one-to-one
  correspondence with) groupoids.
\end{proposition}
\begin{proof}
  By~\cite[Theorem~7]{heunencontrerascattaneo:groupoids}, it suffices
  to show that normalisability implies speciality in $\Rel$. Let
  $(A,\whitemult,\whiteunit,\whitenorm)$ be a normalisable dagger
  Frobenius algebra in $\Rel$. 
  Then the normaliser $\whitenorm$ is an isomorphism. 
  In $\Rel$, this means $\whitenorm = \{ (a,z(a)) \mid a \in A\}$
  for a bijection $z \colon A \to A$. But $\whitenorm$ is also
  positive, and hence self-adjoint. Since all isomorphisms in $\Rel$ are unitary, $z$ equals its own 
  inverse. Therefore $\whitenorm \circ \whitenorm = \id[A]$, that is,
  $(A,\whitemult)$ is special.
\end{proof}

Explicitly, the set $A=\Mor(\cat{G})$ of morphisms of a groupoid
$\cat{G}$ becomes an abstract C*-algebra in $\Rel$ under
\begin{align*}
  \whitemult & = \{ ((g,f),g \circ f) \mid f \text{ and } g \text{ are
    composable morphisms in } \cat{G} \}, \\
  \whiteunit & = \{ (*, \id[a]) \mid a \text{ is an object of } \cat{G} \}.
\end{align*}
The proof of the previous proposition illustrates that we may take
$\whitenorm=\id[A]$, but that normalisers of dagger Frobenius algebras
are not unique.  

Next, we determine the morphisms of $\CPs[\Rel]$. 

\begin{definition}
  A relation $R \subseteq \Mor(\cat{G}) \times
  \Mor(\cat{H})$ between groupoids $\cat{G}$ and $\cat{H}$ 
  \emph{repects inverses} when $(g,h) \in R$ implies $(g^{-1},h^{-1})
  \in R$ and $(\id[\dom(g)],\id[\dom(h)]) \in R$.
\end{definition}

\begin{proposition}\label{prop:cpsrel}
  The category $\CPs[\Rel]$ is isomorphic to the category of
  groupoids and relations respecting inverses.     
\end{proposition}
\begin{proof}
  Unfolding definitions shows that a morphism $R \subseteq (A \times
  A) \times (B \times B)$ in $\Rel$ is completely positive, \ie is of
  the form of the right-hand side of equation~\eqref{eq:cpstar},
  precisely when 
  \begin{equation}
    ((a,a'),(b,b') \in R \implies ((a',a),(b',b) \in R,\,\, ((a,a),(b,b)) \in R.
    \tag{$*$}
  \end{equation}
  If $\cat{G}$ and $\cat{H}$ are groupoids, corresponding to Frobenius
  algebras $(G,\whitemult)$ and $(H,\graymult)$, and $R \subseteq G
  \times H$, then 
  \begin{align*}
    \whitecoaction 
    & = \{ ((g,g'),g^{-1}\circ g') \in G^3 \mid
    g^{-1} \text{ and } g' \text{ are composable} \}, \\
    \grayaction \circ R \circ \whitecoaction 
    & = \{ ((g,g'),(h,h')) \in G^2 \times H^2 \mid 
    g^{-1} \text{ and } g' \text{ are composable} \}, \\
    & \phantom{= \{ ((g,g'),(h,h')) \in G^2 \times H^2 \mid \;\,}
    h^{-1} \text{ and } h' \text{ are composable} \}, \\
    & \phantom{= \{ ((g,g'),(h,h')) \in G^2 \times H^2 \mid \;\,}
    (g^{-1} \circ g', h^{-1} \circ h') \in R \}.
  \end{align*}
  Substituting this into~($*$) translates precisely into $R$
  respecting inverses. 
\end{proof}

Next we investigate ``completely quantum'' objects in
$\CPs[\Rel]$. Recall that a category is \emph{indiscrete} when there
is precisely one morphism between each two objects. Indiscrete
categories are automatically groupoids.  

\begin{proposition}\label{prop:indiscrete}
  The objects in $\CPs[\Rel]$ that are isomorphic to $\B(A)$ for some
  set $A$ are (in one-to-one correspondence with) indiscrete groupoids.
\end{proposition}
\begin{proof}
  By definition, $\B(A)$ corresponds to a groupoid whose set of
  morphisms is $A \times A$, and whose composition is given by
  \[
    (b_2,b_1) \circ (a_2,a_1) = \left\{ \begin{array}{ll} (b_2,a_1) &
        \text{ if } b_1=a_2, \\ \text{undefined} & \text{
          otherwise}. \end{array} \right. 
  \]
  We deduce that the identity morphisms of $\B(A)$ are the pairs
  $(a_2,a_1)$ with $a_2=a_1$. So objects of $\B(A)$ just correspond to
  elements of $A$. Similarly, we find that the morphism $(a_2,a_1)$
  has domain $a_1$ and codomain $a_2$. Hence $(a_2,a_1)$ is the unique
  morphism $a_1 \to a_2$ in $\B(A)$.
\end{proof}

In other words, the essential image of the embedding $\B \colon
\CPM[\Rel] \to \CPs[\Rel]$ is the full subcategory of $\CPs[\Rel]$
consisting of indiscrete groupoids. 

There are many more connections between the theory of groupoids and
abstract C*-algebras. For example, projections in an abstract
C*-algebra in $\Rel$ are precisely the connected components of its
corresponding
groupoid~\cite[Lemma~22]{coeckeheunenkissinger:compositional}.

\subsection{Positive matrices}

To conclude this section, we consider categories that are in some sense
between the categories $\FHilb$ and $\Rel$; the former can be thought
of as involving matrices over the complex numbers, whereas the latter
can be thought of as involving matrices over the two element set. We
will consider matrices ranging over other domains.

We start with the category $\MatR$. Its
objects are natural numbers, and a morphism $m \to n$ is an $m$-by-$n$
matrix whose entries are nonnegative real numbers, \ie elements of
$[0,\infty)$. Composition is matrix multiplication, and identity
matrices give identity morphisms. Tensor product acts as
multiplication on objects, and as Kronecker product on morphisms.

We will determine the objects of $\CPs[\MatR]$ by reducing to
$\CPs[\FHilb]$. There is an obvious dagger symmetric monoidal
functor $\MatR \to \FHilb$, sending $n$ to $\C^n$ with its canonical
basis. Hence a normalisable dagger Frobenius algebra
$(n,\whitemult,\whiteunit,\whitenorm)$ in $\MatR$ also defines a
C*-algebra structure on $\C^n$. 
Recall that any finite-dimensional C*-algebra $A$ can be written in
standard form as $A\cong\bigoplus_{k} \M_{n_k}$. 

\begin{definition}
  Write $E_n=\{e_{ij} \mid i,j=1,\ldots,n\}$ for the standard basis of $\M_n$.
  By the \emph{matrix} of a linear map $f \colon \M_m \to \M_n$, we mean the function $F \colon E_m \times E_n \to \mathbb{C}$ given by the entries $F(e_{ij},e_{kl}) = \langle e_{kl} \mid f\big(e_{ij}\big) \rangle = \Tr(e_{lk} f(e_{ij}))$.
  This definition extends to linear maps $f \colon \bigoplus_k \M_{m_k} \to \bigoplus_l \M_{n_l}$ between finite-dimensional C*-algebras in standard form. 
  We say that $f$ is \emph{really positive} when its matrix $F$ has entries in $\Rpos$.
  If $f$ is completely positive and really positive, we call it \emph{really completely positive}.
\end{definition}

\begin{proposition}\label{prop:matr}
  The category $\CPs[\MatR]$ is isomorphic to the category of
  finite-dimensional C*-algebras in standard form and really completely positive maps.
\end{proposition}
\begin{proof}
  The fact that the functor $\MatR \to \FHilb$ is dagger symmetric
  monoidal and faithful implies that the induced functor
  $\CPs[\MatR] \to \CPs[\FHilb]$ is also dagger symmetric monoidal and faithful.
  It is full by construction, and injective on objects.
  Hence it suffices to show that it is surjective on objects.
  First, observe that the structure maps
  $\whitemult$, $\whiteunit$, and $\whitenorm$ of the C*-algebra $\M_n$ are really completely positive.
  The matrices $M \colon E_n^3 \to \Rpos$ for multiplication, $U \colon E_n \to \Rpos$ for the unit, and $Z
  \colon E_n^2 \to \Rpos$ then take the form
  \begin{align*}
    M(e_{ij},e_{kl},e_{pq}) & = \delta_{jk} \delta_{ip} \delta _{lq},  \\
    U(e_{ij}) & = \delta_{ii}, \\
    N(e_{ij}) & = 1/\sqrt{n}.
  \end{align*}
  Hence $\M_n$ is in the image of the functor $\CPs[\MatR] \to
  \CPs[\FHilb]$. Because $\CPs[\MatR]$ has biproducts, C*-algebras in
  standard form are reached, too.
  %
\end{proof}

Finally, let us consider matrices with entries ranging over other
sets of positive numbers, such as the unit interval $[0,1]$.
To be precise, we will consider the category $\Mat(Q)$, where $Q$ is a
cancellative commutative quantale. Recall that a \emph{quantale} is a
partial order $(Q,\leq)$ that has suprema of arbitrary subsets,
together with a commutative multiplication $(Q,\cdot,1)$ satisfying
\[
  x \cdot (\bigvee_i y_i) = \bigvee_i x \cdot y_i.
\]
It is cancellative when $x \cdot y = x \cdot z$ implies $y=z$ or
$x=0$, where $0 = \bigvee
\emptyset$. For more information we refer to~\cite{rosenthal:quantales}.
The extended nonnegative real numbers $[0,\infty]$ form an example
under the usual ordering and multiplication, as does the unit interval
$[0,1]$. Another example is the Boolean algebra $\{0,1\}$ under the
usual ordering and multiplication. 

The category $\Mat(Q)$ has sets as objects, morphisms $A \to B$
are $Q$-valued matrices, \ie functions $A \times B \to Q$. Composition
of $R \colon A \to B$ and $S \colon B \to C$ is 
\[
  S \circ R(a,c) = \bigvee_b R(a,b) \cdot S(b,c).
\]
Cartesian product and matrix transpose makes this into a dagger
compact category, very much like $\Rel$. In fact, notice that
$\Mat(\{0,1\}) = \Rel$.\footnote{Notice also that $\mathbb{R}_{\geq 0}$ is not a quantale under its usual ordering.} 

\begin{lemma}
  Any normalisable dagger Frobenius algebra in $\Mat(Q)$ induces a
  groupoid. 
\end{lemma}
\begin{proof}
  Let $(A,\whitemult,\whiteunit,\whitenorm)$ be a normalisable dagger
  Frobenius algebra in $\Mat(Q)$.
  Notice that there is a (unique) homomorphism $f \colon Q \to
  \{0,1\}$ of quantales such that $f(x)=0$ if and only if 
  $x=0$, namely 
  \[
    f(x) = \left\{ \begin{array}{ll} 0 & \text{ if } x = 0, \\ 1 &
        \text{ otherwise.} \end{array}\right.
  \]
  It induces a dagger symmetric monoidal functor $f^* \colon \Mat(Q) \to
  \Rel$; see also~\cite[Section~5.2]{abramskyheunen:hstar}.
  Therefore $G:=A$ becomes a dagger
  Frobenius algebra in $\Rel$ with multiplication $f^*(\whitemult)$
  and unit $f^*(\whiteunit)$. Moreover, $f^*(\whitecounit) = f^*(\whiteloop)
  \circ f^*(\whitenorm)^2$ by normalisability. But, as in the proof of 
  Proposition~\ref{prop:cpstarrelobjs}, $f^*(\whitenorm)$ is a
  positive isomorphism in $\Rel$, and so $f^*(\whitenorm)^2=\id[G]$.
  At this point Lemma~\ref{lem:normalisability} guarantees that
  $(G,f^*(\whitenorm),f^*(\whiteunit),\id[G])$
  is a normalisable dagger Frobenius algebra in $\Rel$, which
  corresponds to a groupoid by Proposition~\ref{prop:cpstarrelobjs}.
\end{proof}

The previous lemma shows that if we ``collapse'' the matrix $M \colon G^3 \to Q$ of
multiplication to $M \colon G^3 \to \{0,1\}$, it becomes the
multiplication table of a groupoid. Similarly, the matrix $U \colon G
\to Q$ becomes the set of identities of that groupoid. The only freedom
left is what nonzero elements of $Q$ to place in the nonzero entries of
these matrices. It is easy to obtain several constraints on these values~\cite[Section~5.2]{abramskyheunen:hstar}.
However, in general, $\CPs[\Mat(Q)]$ does not seem to correspond to a
familiar category such as $\CPs[\Rel]$. We refrain from
explicating it further, but note that it does provide a nonstandard
model that lends itself to easy calculation, for example to find
counterexamples.

\bibliographystyle{spmpsci}      
\bibliography{channels}

\end{document}

%% file: figures/cap_cup.tikz
\begin{tikzpicture}[dotpic]
	\begin{pgfonlayer}{nodelayer}
		\node [style=white dot] (0) at (-5.25, -0.25) {};
		\node [style=white dot] (1) at (-2.75, -0) {};
		\node [style=white dot] (2) at (-2.75, -0.75) {};
		\node [style=none] (3) at (-5.75, 0.5) {};
		\node [style=none] (4) at (-4.75, 0.5) {};
		\node [style=none] (5) at (-3.25, 0.75) {};
		\node [style=none] (6) at (-2.25, 0.75) {};
		\node [style=none] (7) at (-4, -0) {$:=$};
		\node [style=none] (8) at (3.25, -0.75) {};
		\node [style=white dot] (9) at (2.75, 0) {};
		\node [style=none] (10) at (2.25, -0.75) {};
		\node [style=none] (11) at (1.5, 0) {$:=$};
		\node [style=white dot] (12) at (0.25, 0.25) {};
		\node [style=white dot] (13) at (2.75, 0.75) {};
		\node [style=none] (14) at (0.75, -0.5) {};
		\node [style=none] (15) at (-0.25, -0.5) {};
		\node [style=none] (16) at (5.5, -0) {};
		\node [style=white dot] (17) at (6, -0.75) {};
		\node [style=none] (18) at (7.5, -0) {};
		\node [style=white dot] (19) at (7, 0.75) {};
		\node [style=none] (20) at (6.5, -0) {};
		\node [style=none] (21) at (5.5, 0.75) {};
		\node [style=none] (22) at (7.5, -0.75) {};
		\node [style=none] (23) at (9, 0.75) {};
		\node [style=none] (24) at (9, -0.75) {};
		\node [style=none] (25) at (8.25, -0) {$=$};
		\node [style=white dot] (26) at (12, -0.75) {};
		\node [style=none] (27) at (11.5, -0) {};
		\node [style=white dot] (28) at (11, 0.75) {};
		\node [style=none] (29) at (12.5, -0) {};
		\node [style=none] (30) at (10.5, -0) {};
		\node [style=none] (31) at (10.5, -0.75) {};
		\node [style=none] (32) at (12.5, 0.75) {};
		\node [style=none] (33) at (9.75, -0) {$=$};
	\end{pgfonlayer}
	\begin{pgfonlayer}{edgelayer}
		\draw [style=diredge, in=-90, out=165, looseness=1.25] (0) to (3.center);
		\draw [style=diredge, in=-90, out=15, looseness=1.25] (0) to (4.center);
		\draw [style=diredge] (2) to (1);
		\draw [style=diredge, bend left=15] (1) to (5.center);
		\draw [style=diredge, bend right=15] (1) to (6.center);
		\draw [style=diredge, in=-165, out=90, looseness=1.25] (15.center) to (12);
		\draw [style=diredge, in=-15, out=90, looseness=1.25] (14.center) to (12);
		\draw [style=diredge] (9) to (13);
		\draw [style=diredge, bend left=15] (10.center) to (9);
		\draw [style=diredge, bend right=15] (8.center) to (9);
		\draw [in=-90, out=165, looseness=1.25] (17) to (16.center);
		\draw [style=diredge, in=-165, out=90, looseness=1.25] (20.center) to (19);
		\draw [style=diredge, in=-15, out=90, looseness=1.25] (18.center) to (19);
		\draw [style=diredge] (16.center) to (21.center);
		\draw (22.center) to (18.center);
		\draw [style=diredge] (24.center) to (23.center);
		\draw [in=-90, out=15, looseness=1.25] (17) to (20.center);
		\draw [in=-90, out=15, looseness=1.25] (26) to (29.center);
		\draw [style=diredge, in=-15, out=90, looseness=1.25] (27.center) to (28);
		\draw [style=diredge, in=-165, out=90, looseness=1.25] (30.center) to (28);
		\draw [style=diredge] (29.center) to (32.center);
		\draw (31.center) to (30.center);
		\draw [in=-90, out=165, looseness=1.25] (26) to (27.center);
	\end{pgfonlayer}
\end{tikzpicture}

%% file: figures/sfa_ident.tikz
\begin{tikzpicture}[dotpic]
	\begin{pgfonlayer}{nodelayer}
		\node [style=white dot] (0) at (-1.25, 0.5) {};
		\node [style=white dot] (1) at (-1.25, -0.25) {};
		\node [style=none] (2) at (-1.75, -0.75) {};
		\node [style=none] (3) at (-0.75, -0.75) {};
		\node [style=none] (4) at (-1.75, 1) {};
		\node [style=none] (5) at (-0.75, 1) {};
		\node [style=none] (6) at (2.5, 1) {};
		\node [style=none] (7) at (2.5, -0.75) {};
		\node [style=white dot] (8) at (1.25, -0) {};
		\node [style=none] (9) at (1.25, -0.75) {};
		\node [style=white dot] (10) at (2.5, -0) {};
		\node [style=none] (11) at (1.25, 1) {};
		\node [style=none] (12) at (0, -0) {$=$};
	\end{pgfonlayer}
	\begin{pgfonlayer}{edgelayer}
		\draw [style=diredge, bend right=15] (4.center) to (0);
		\draw [style=diredge, bend right=15] (0) to (5.center);
		\draw [style=diredge, bend right=15] (3.center) to (1);
		\draw [style=diredge, bend right=15] (1) to (2.center);
		\draw [style=diredge] (1) to (0);
		\draw [style=diredge] (11.center) to (8);
		\draw [style=diredge] (8) to (9.center);
		\draw [style=diredge] (7.center) to (10);
		\draw [style=diredge] (10) to (6.center);
		\draw [style=diredge, bend right=60, looseness=1.25] (10) to (8);
	\end{pgfonlayer}
\end{tikzpicture}

%% file: figures/sfa_ident_pf.tikz
\begin{tikzpicture}[dotpic]
	\begin{pgfonlayer}{nodelayer}
		\node [style=white dot] (0) at (-8.5, 0.5) {};
		\node [style=white dot] (1) at (-8.5, -0.5) {};
		\node [style=none] (2) at (-9, -1.25) {};
		\node [style=none] (3) at (-8, -1.25) {};
		\node [style=none] (4) at (-9, 1.25) {};
		\node [style=none] (5) at (-8, 1.25) {};
		\node [style=none] (6) at (-7.5, -0) {$=$};
		\node [style=none] (7) at (-6.75, -1.25) {};
		\node [style=none] (8) at (-6, -1.25) {};
		\node [style=white dot] (9) at (-6, 0.5) {};
		\node [style=white dot] (10) at (-6, -0.5) {};
		\node [style=none] (11) at (-6, 1.25) {};
		\node [style=none] (12) at (-6.75, 1.25) {};
		\node [style=none] (13) at (-5, -0) {$=$};
		\node [style=white dot] (14) at (-3, -0.75) {};
		\node [style=none] (15) at (-3.75, 1.5) {};
		\node [style=none] (16) at (-4.5, -1.5) {};
		\node [style=none] (17) at (-2, 1.5) {};
		\node [style=none] (18) at (-3, -1.5) {};
		\node [style=white dot] (19) at (-2.5, 0.5) {};
		\node [style=white dot] (20) at (-3, 1.25) {};
		\node [style=none] (21) at (-4, 0.25) {};
		\node [style=none] (22) at (-0.5, 1.5) {};
		\node [style=none] (23) at (-1.25, -1.5) {};
		\node [style=none] (24) at (1.25, 1.5) {};
		\node [style=none] (25) at (-0.75, 0.5) {};
		\node [style=none] (26) at (-1.75, -0) {$=$};
		\node [style=white dot] (27) at (0.5, -0.75) {};
		\node [style=white dot] (28) at (0, -0) {};
		\node [style=none] (29) at (0.5, -1.5) {};
		\node [style=white dot] (30) at (0.25, 1.25) {};
		\node [style=white dot] (31) at (3.25, -0) {};
		\node [style=none] (32) at (1.75, -0) {$=$};
		\node [style=none] (33) at (2.75, 0.5) {};
		\node [style=white dot] (34) at (3.75, -0.75) {};
		\node [style=none] (35) at (4.5, 1.5) {};
		\node [style=none] (36) at (5, 1.5) {};
		\node [style=white dot] (37) at (3.75, 1.25) {};
		\node [style=none] (38) at (2.25, -1.5) {};
		\node [style=none] (39) at (3.75, -1.5) {};
		\node [style=none] (40) at (7.5, 1.5) {};
		\node [style=white dot] (41) at (7.25, -0.75) {};
		\node [style=none] (42) at (5.75, -1.5) {};
		\node [style=none] (43) at (8.5, 1.5) {};
		\node [style=white dot] (44) at (6.75, 0.5) {};
		\node [style=none] (45) at (5.25, -0) {$=$};
		\node [style=none] (46) at (6.25, 1) {};
		\node [style=none] (47) at (7.25, -1.5) {};
		\node [style=none] (48) at (8.75, -0) {$=$};
		\node [style=none] (49) at (10.75, -1.25) {};
		\node [style=none] (50) at (9.5, 1.25) {};
		\node [style=none] (51) at (10.75, 1.25) {};
		\node [style=white dot] (52) at (9.5, -0) {};
		\node [style=none] (53) at (9.5, -1.25) {};
		\node [style=white dot] (54) at (10.75, -0) {};
	\end{pgfonlayer}
	\begin{pgfonlayer}{edgelayer}
		\draw [style=diredge, bend right=15] (4.center) to (0);
		\draw [style=diredge, bend right=15] (0) to (5.center);
		\draw [style=diredge, bend right=15] (3.center) to (1);
		\draw [style=diredge, bend right=15] (1) to (2.center);
		\draw [style=diredge] (1) to (0);
		\draw [style=diredge, in=-135, out=-90, looseness=1.50] (12.center) to (9);
		\draw [style=diredge] (9) to (11.center);
		\draw [style=diredge] (8.center) to (10);
		\draw [style=diredge, in=90, out=135, looseness=1.50] (10) to (7.center);
		\draw [style=diredge] (10) to (9);
		\draw [style=diredge, bend right=15] (19) to (17.center);
		\draw [style=diredge] (18.center) to (14);
		\draw [style=diredge] (14) to (19);
		\draw [style=diredge, in=-120, out=-90, looseness=2.50] (15.center) to (20);
		\draw [in=0, out=135, looseness=0.75] (14) to (21.center);
		\draw [style=diredge, in=90, out=180, looseness=0.75] (21.center) to (16.center);
		\draw [style=diredge] (19) to (20);
		\draw [style=diredge, in=-120, out=-90, looseness=2.50] (22.center) to (30);
		\draw [in=0, out=135, looseness=0.75] (28) to (25.center);
		\draw [style=diredge, in=90, out=180, looseness=0.75] (25.center) to (23.center);
		\draw [style=diredge, in=-60, out=45] (28) to (30);
		\draw [style=diredge] (29.center) to (27);
		\draw [style=diredge] (27) to (28);
		\draw [style=diredge, bend right=15] (27) to (24.center);
		\draw [style=diredge, in=-60, out=-90, looseness=2.50] (35.center) to (37);
		\draw [in=0, out=135, looseness=0.75] (31) to (33.center);
		\draw [style=diredge, in=90, out=180, looseness=0.75] (33.center) to (38.center);
		\draw [style=diredge, in=-135, out=45] (31) to (37);
		\draw [style=diredge] (39.center) to (34);
		\draw [style=diredge] (34) to (31);
		\draw [style=diredge, bend right=15] (34) to (36.center);
		\draw [in=0, out=135, looseness=0.75] (44) to (46.center);
		\draw [style=diredge, in=90, out=180, looseness=0.50] (46.center) to (42.center);
		\draw [style=diredge] (47.center) to (41);
		\draw [style=diredge, in=-135, out=142, looseness=1.25] (41) to (44);
		\draw [style=diredge, bend right=15] (41) to (43.center);
		\draw [style=diredge, in=-45, out=-90, looseness=1.50] (40.center) to (44);
		\draw [style=diredge] (50.center) to (52);
		\draw [style=diredge] (52) to (53.center);
		\draw [style=diredge] (49.center) to (54);
		\draw [style=diredge] (54) to (51.center);
		\draw [style=diredge, bend right=60, looseness=1.25] (54) to (52);
	\end{pgfonlayer}
\end{tikzpicture}

%% file: figures/norm_alt.tikz
\begin{tikzpicture}[dotpic]
	\begin{pgfonlayer}{nodelayer}
		\node [style=none] (0) at (1, -1.25) {};
		\node [style=none] (1) at (1, 1.25) {};
		\node [style=none] (2) at (-1.5, 1.5) {};
		\node [style=none] (3) at (-1.5, -1.5) {};
		\node [style=none] (4) at (0, -0) {$=$};
		\node [style=white dot] (5) at (-1.5, 0.75) {};
		\node [style=white dot] (6) at (-1.5, -0.75) {};
		\node [style=white norm] (7) at (-2, -0) {};
		\node [style=white norm] (8) at (-1, -0) {};
	\end{pgfonlayer}
	\begin{pgfonlayer}{edgelayer}
		\draw [style=diredge] (0.center) to (1.center);
		\draw [style=diredge] (5) to (2.center);
		\draw [style=diredge] (3.center) to (6);
		\draw [style=diredge, in=-90, out=30] (6) to (8);
		\draw [style=diredge, in=-30, out=90] (8) to (5);
		\draw [style=diredge, in=90, out=-150] (5) to (7);
		\draw [style=diredge, in=150, out=-90] (7) to (6);
	\end{pgfonlayer}
\end{tikzpicture}

%% file: figures/norm_alt_pf.tikz
\begin{tikzpicture}[dotpic]
	\begin{pgfonlayer}{nodelayer}
		\node [style=none] (0) at (-1.5, 1.75) {};
		\node [style=none] (1) at (-1.5, -1.75) {};
		\node [style=none] (2) at (0, -0) {$=$};
		\node [style=white dot] (3) at (-1.5, 1) {};
		\node [style=white dot] (4) at (-1.5, -1) {};
		\node [style=white norm] (5) at (-2, -0) {};
		\node [style=white norm] (6) at (-1, -0) {};
		\node [style=white dot] (7) at (1.25, 1.25) {};
		\node [style=none] (8) at (1.25, -2) {};
		\node [style=white dot] (9) at (1.25, -0) {};
		\node [style=none] (10) at (1.25, 2) {};
		\node [style=white norm] (11) at (1.25, -1.25) {};
		\node [style=white norm] (12) at (1.25, -0.75) {};
		\node [style=none] (13) at (4.5, -2) {};
		\node [style=white dot] (14) at (4.25, 1) {};
		\node [style=none] (15) at (4.75, 2) {};
		\node [style=white norm] (16) at (4.5, -0.75) {};
		\node [style=white norm] (17) at (4.5, -1.25) {};
		\node [style=white dot] (18) at (4.25, -0) {};
		\node [style=none] (19) at (2.5, -0) {$=$};
		\node [style=none] (20) at (5.75, -0) {$=$};
		\node [style=white dot] (21) at (8, -0) {};
		\node [style=white norm] (22) at (8, -0.75) {};
		\node [style=none] (23) at (8.75, 2) {};
		\node [style=none] (24) at (8, -2) {};
		\node [style=white norm] (25) at (8, -1.25) {};
		\node [style=none] (26) at (6.75, 0.5) {};
		\node [style=white dot] (27) at (7.25, 1) {};
		\node [style=none] (28) at (6.75, 1.75) {};
		\node [style=none] (29) at (9.5, -0) {$=$};
		\node [style=none] (30) at (10.5, 1.75) {};
		\node [style=white norm] (31) at (11.25, -0.25) {};
		\node [style=none] (32) at (10.5, 0.5) {};
		\node [style=white norm] (33) at (11.25, 0.25) {};
		\node [style=none] (34) at (12.5, 2) {};
		\node [style=white dot] (35) at (11.75, -1) {};
		\node [style=none] (36) at (11.75, -2) {};
		\node [style=white dot] (37) at (11, 1) {};
		\node [style=none] (38) at (16.75, -0) {$=$};
		\node [style=white dot] (39) at (17.5, 0.5) {};
		\node [style=white dot] (40) at (18, -0.75) {};
		\node [style=none] (41) at (18.75, 1.75) {};
		\node [style=none] (42) at (18, -1.75) {};
		\node [style=none] (43) at (19.25, -0) {$=$};
		\node [style=none] (44) at (20.25, -1.25) {};
		\node [style=none] (45) at (20.25, 1.25) {};
		\node [style=none] (46) at (3.75, 1.5) {};
		\node [style=none] (47) at (3.75, -0.5) {};
		\node [style=none] (48) at (13.25, -0) {$=$};
		\node [style=white norm] (49) at (14.25, 0.5) {};
		\node [style=none] (50) at (16.25, 2) {};
		\node [style=white dot] (51) at (15.25, -1) {};
		\node [style=none] (52) at (15.25, -2) {};
		\node [style=none] (53) at (15, 2) {};
		\node [style=white dot] (54) at (14.5, 1.25) {};
		\node [style=none] (55) at (15, 0.75) {};
		\node [style=white norm] (56) at (14.25, -0) {};
		\node [style=none] (57) at (13.25, 0.75) {\small $(*)$};
	\end{pgfonlayer}
	\begin{pgfonlayer}{edgelayer}
		\draw [style=diredge] (3) to (0.center);
		\draw [style=diredge] (1.center) to (4);
		\draw [style=diredge, in=-90, out=30] (4) to (6);
		\draw [style=diredge, in=-30, out=90] (6) to (3);
		\draw [style=diredge, in=90, out=-150] (3) to (5);
		\draw [style=diredge, in=150, out=-90] (5) to (4);
		\draw [style=diredge] (7) to (10.center);
		\draw [style=diredge] (8.center) to (11);
		\draw (11) to (12);
		\draw [style=diredge] (12) to (9);
		\draw [style=diredge, bend right=45, looseness=1.25] (9) to (7);
		\draw [style=diredge, bend right=45, looseness=1.25] (7) to (9);
		\draw [style=diredge] (14) to (15.center);
		\draw [style=diredge] (13.center) to (17);
		\draw (17) to (16);
		\draw [style=diredge, in=-72, out=90] (16) to (18);
		\draw [style=diredge] (18) to (14);
		\draw [style=diredge, bend right=15] (21) to (23.center);
		\draw [style=diredge] (24.center) to (25);
		\draw (25) to (22);
		\draw [style=diredge] (22) to (21);
		\draw [in=0, out=90] (27) to (28.center);
		\draw [style=diredge, in=-105, out=0] (26.center) to (27);
		\draw [in=180, out=180, looseness=1.50] (28.center) to (26.center);
		\draw [style=diredge, in=-45, out=135] (21) to (27);
		\draw [style=diredge, in=-90, out=45, looseness=0.75] (35) to (34.center);
		\draw (31) to (33);
		\draw [in=0, out=90] (37) to (30.center);
		\draw [style=diredge, in=-105, out=0] (32.center) to (37);
		\draw [in=180, out=180, looseness=1.50] (30.center) to (32.center);
		\draw [style=diredge, in=-90, out=135] (35) to (31);
		\draw [style=diredge, in=-60, out=90] (33) to (37);
		\draw [style=diredge] (36.center) to (35);
		\draw [style=diredge, in=-90, out=45, looseness=0.75] (40) to (41.center);
		\draw [style=diredge] (42.center) to (40);
		\draw [style=diredge, in=-90, out=142] (40) to (39);
		\draw [style=diredge] (44.center) to (45.center);
		\draw [in=0, out=120] (14) to (46.center);
		\draw [in=180, out=180] (46.center) to (47.center);
		\draw [style=diredge, in=-120, out=0] (47.center) to (18);
		\draw [style=diredge, in=-90, out=45, looseness=0.75] (51) to (50.center);
		\draw (56) to (49);
		\draw [in=180, out=90] (54) to (53.center);
		\draw [style=diredge, in=-75, out=180] (55.center) to (54);
		\draw [in=0, out=0, looseness=1.50] (53.center) to (55.center);
		\draw [style=diredge, in=-90, out=135] (51) to (56);
		\draw [style=diredge, in=-120, out=90] (49) to (54);
		\draw [style=diredge] (52.center) to (51);
	\end{pgfonlayer}
\end{tikzpicture}

%% file: figures/cpstar_condition.tikz
\begin{tikzpicture}[dotpic]
	\begin{pgfonlayer}{nodelayer}
		\node [style=none] (0) at (0, -0) {$=$};
		\node [style=square box] (1) at (-2, -0) {$f$};
		\node [style=gray dot] (2) at (-2, 1) {};
		\node [style=none] (3) at (-1.25, 2) {};
		\node [style=none] (4) at (-2.75, 2) {};
		\node [style=white dot] (5) at (-2, -1) {};
		\node [style=none] (6) at (-1.25, -2) {};
		\node [style=none] (7) at (-2.75, -2) {};
		\node [style=square box, minimum width=1 cm, minimum height=0.75 cm] (8) at (2, -0) {$g_*$};
		\node [style=none] (9) at (4.5, -0.75) {};
		\node [style=none] (10) at (2, -0.75) {};
		\node [style=none] (11) at (2, -1.75) {};
		\node [style=none] (12) at (4.5, -1.75) {};
		\node [style=none] (13) at (5, 0.75) {};
		\node [style=none] (14) at (1.5, 0.75) {};
		\node [style=none] (15) at (5, 2) {};
		\node [style=none] (16) at (1.5, 2) {};
		\node [style=square box, minimum width=1 cm, minimum height=0.75 cm] (17) at (4.5, -0) {$g$};
		\node [style=none] (18) at (2.5, 0.75) {};
		\node [style=none] (19) at (4, 0.75) {};
	\end{pgfonlayer}
	\begin{pgfonlayer}{edgelayer}
		\draw [style=diredge] (1) to (2);
		\draw [style=diredge, bend right] (4.center) to (2);
		\draw [style=diredge, bend right] (2) to (3.center);
		\draw [style=diredge, bend right] (5) to (7.center);
		\draw [style=diredge, bend right] (6.center) to (5);
		\draw [style=diredge] (5) to (1);
		\draw [style=diredge] (12.center) to (9.center);
		\draw [style=diredge] (13.center) to (15.center);
		\draw [style=diredge] (10.center) to (11.center);
		\draw [style=diredge] (16.center) to (14.center);
		\draw [style=diredge, in=90, out=90, looseness=1.50] (19.center) to (18.center);
	\end{pgfonlayer}
\end{tikzpicture}

%% file: figures/cpstar_conv_form.tikz
\begin{tikzpicture}[dotpic]
	\begin{pgfonlayer}{nodelayer}
		\node [style=none] (0) at (2, -0.75) {};
		\node [style=square box, minimum width=1 cm, minimum height=0.75 cm] (1) at (4.5, 0) {$h$};
		\node [style=none] (2) at (0, 0) {$=$};
		\node [style=square box, minimum width=1 cm, minimum height=0.75 cm] (3) at (2, 0) {$h_*$};
		\node [style=none] (4) at (5, 0.75) {};
		\node [style=none] (5) at (4, 0.75) {};
		\node [style=none] (6) at (-2, 1.75) {};
		\node [style=none] (7) at (4.5, -0.75) {};
		\node [style=square box] (8) at (-2, 0) {$f$};
		\node [style=none] (9) at (1.5, 0.75) {};
		\node [style=none] (10) at (-2, -1.75) {};
		\node [style=none] (11) at (2.5, 0.75) {};
		\node [style=gray dot] (12) at (3.25, 2) {};
		\node [style=white dot] (13) at (3.25, -1.75) {};
		\node [style=none] (14) at (3.25, 2.75) {};
		\node [style=none] (15) at (3.25, -2.5) {};
	\end{pgfonlayer}
	\begin{pgfonlayer}{edgelayer}
		\draw [style=diredge] (8) to (6.center);
		\draw [style=diredge] (10.center) to (8);
		\draw [style=diredge, in=90, out=90, looseness=1.50] (5.center) to (11.center);
		\draw [style=diredge, in=0, out=90] (4.center) to (12);
		\draw [style=diredge, in=90, out=180] (12) to (9.center);
		\draw [style=diredge, in=180, out=-90] (0.center) to (13);
		\draw [style=diredge, in=-90, out=0] (13) to (7.center);
		\draw [style=diredge] (15.center) to (13);
		\draw [style=diredge] (12) to (14.center);
	\end{pgfonlayer}
\end{tikzpicture}

%% file: figures/cpstar_dagger.tikz
\begin{tikzpicture}[dotpic]
	\begin{pgfonlayer}{nodelayer}
		\node [style=none] (0) at (0, -0) {$=$};
		\node [style=square box] (1) at (-2, -0) {$f^\dagger$};
		\node [style=white dot] (2) at (-2, 1.25) {};
		\node [style=none] (3) at (-1.25, 2.25) {};
		\node [style=none] (4) at (-2.75, 2.25) {};
		\node [style=gray dot] (5) at (-2, -1.25) {};
		\node [style=none] (6) at (-1.25, -2.25) {};
		\node [style=none] (7) at (-2.75, -2.25) {};
		\node [style=square box, minimum width=1 cm, minimum height=0.75 cm] (8) at (2, -0) {$g^*$};
		\node [style=none] (9) at (5, -0.75) {};
		\node [style=none] (10) at (1.5, -0.75) {};
		\node [style=none] (11) at (1.5, -2.25) {};
		\node [style=none] (12) at (5, -2.25) {};
		\node [style=none] (13) at (4.5, 0.75) {};
		\node [style=none] (14) at (2, 0.75) {};
		\node [style=none] (15) at (4.5, 2.25) {};
		\node [style=none] (16) at (2, 2.25) {};
		\node [style=square box, minimum width=1 cm, minimum height=0.75 cm] (17) at (4.5, -0) {$g^\dagger$};
		\node [style=none] (18) at (4, -0.75) {};
		\node [style=none] (19) at (2.5, -0.75) {};
		\node [style=none] (20) at (11.25, -0.5) {};
		\node [style=none] (21) at (7.75, -2.5) {};
		\node [style=none] (22) at (13.25, -0.5) {};
		\node [style=none] (23) at (13.25, -2.5) {};
		\node [style=none] (24) at (6.25, -0) {$=$};
		\node [style=none] (25) at (12.75, 2.5) {};
		\node [style=none] (26) at (8.25, 2.5) {};
		\node [style=square box, minimum width=1 cm, minimum height=0.75 cm] (27) at (8.25, 0.25) {$g^*$};
		\node [style=none] (28) at (7.75, -0.5) {};
		\node [style=none] (29) at (8.25, 1) {};
		\node [style=square box, minimum width=1 cm, minimum height=0.75 cm] (30) at (12.75, 0.25) {$g^\dagger$};
		\node [style=none] (31) at (12.75, 1) {};
		\node [style=none] (32) at (12.25, -0.5) {};
		\node [style=none] (33) at (8.75, -0.5) {};
		\node [style=none] (34) at (9.75, -0.5) {};
		\node [style=none] (35) at (9.75, 1.75) {};
		\node [style=none] (36) at (11.25, 1.75) {};
		\node [style=none] (37) at (10.25, -1.5) {};
		\node [style=none] (38) at (10.25, 1.5) {};
		\node [style=none] (39) at (7, -1.5) {};
		\node [style=none] (40) at (7, 1.5) {};
		\node [style=none] (41) at (14, 1.5) {};
		\node [style=none] (42) at (10.75, 1.5) {};
		\node [style=none] (43) at (14, -1.5) {};
		\node [style=none] (44) at (10.75, -1.5) {};
		\node [style=none] (45) at (12.25, -2) {$h$};
		\node [style=none] (46) at (9, -2) {$h_*$};
	\end{pgfonlayer}
	\begin{pgfonlayer}{edgelayer}
		\draw [style=diredge] (1) to (2);
		\draw [style=diredge, bend right] (4.center) to (2);
		\draw [style=diredge, bend right] (2) to (3.center);
		\draw [style=diredge, bend right] (5) to (7.center);
		\draw [style=diredge, bend right] (6.center) to (5);
		\draw [style=diredge] (5) to (1);
		\draw [style=diredge] (12.center) to (9.center);
		\draw [style=diredge] (13.center) to (15.center);
		\draw [style=diredge] (10.center) to (11.center);
		\draw [style=diredge] (16.center) to (14.center);
		\draw [style=diredge, in=-90, out=-90, looseness=1.50] (19.center) to (18.center);
		\draw [style=diredge] (23.center) to (22.center);
		\draw [style=diredge] (31.center) to (25.center);
		\draw [style=diredge] (28.center) to (21.center);
		\draw [style=diredge] (26.center) to (29.center);
		\draw [style=diredge, in=-90, out=-90, looseness=1.50] (20.center) to (32.center);
		\draw [in=-90, out=-90, looseness=1.50] (33.center) to (34.center);
		\draw (34.center) to (35.center);
		\draw [in=90, out=90, looseness=2.00] (35.center) to (36.center);
		\draw (36.center) to (20.center);
		\draw [style=dashed edge] (40.center) to (38.center);
		\draw [style=dashed edge] (38.center) to (37.center);
		\draw [style=dashed edge] (37.center) to (39.center);
		\draw [style=dashed edge] (39.center) to (40.center);
		\draw [style=dashed edge] (42.center) to (41.center);
		\draw [style=dashed edge] (41.center) to (43.center);
		\draw [style=dashed edge] (43.center) to (44.center);
		\draw [style=dashed edge] (44.center) to (42.center);
	\end{pgfonlayer}
\end{tikzpicture}

%% file: figures/cpstar_alt_pf.tikz
\begin{tikzpicture}[dotpic]
	\begin{pgfonlayer}{nodelayer}
		\node [style=none] (0) at (-7.5, -0) {$=$};
		\node [style=square box] (1) at (-10.25, -0.5) {$f$};
		\node [style=gray dot] (2) at (-10.25, 1.5) {};
		\node [style=none] (3) at (-9.5, 2.5) {};
		\node [style=none] (4) at (-11, 2.5) {};
		\node [style=white dot] (5) at (-9, -0) {};
		\node [style=none] (6) at (-11.75, 1.75) {};
		\node [style=none] (7) at (-8.75, 2.5) {};
		\node [style=square box, minimum width=1 cm, minimum height=0.75 cm] (8) at (-5.5, -0.25) {$g_*$};
		\node [style=none] (9) at (-2.75, 0.5) {};
		\node [style=none] (10) at (-6.25, 0.5) {};
		\node [style=none] (11) at (-2.75, 1.75) {};
		\node [style=none] (12) at (-6.25, 1.75) {};
		\node [style=square box, minimum width=1 cm, minimum height=0.75 cm] (13) at (-3, -0.25) {$g$};
		\node [style=none] (14) at (-5, 0.5) {};
		\node [style=none] (15) at (-3.5, 0.5) {};
		\node [style=none] (16) at (-10.25, -1.5) {};
		\node [style=none] (17) at (-9, -1.5) {};
		\node [style=none] (18) at (-11.75, 2.5) {};
		\node [style=none] (19) at (-2.25, 0.5) {};
		\node [style=none] (20) at (-2.25, 1.75) {};
		\node [style=none] (21) at (-5.75, 1.75) {};
		\node [style=none] (22) at (-5.75, 0.5) {};
		\node [style=none] (23) at (0, -0) {$\Rightarrow$};
		\node [style=none] (24) at (4.25, 2) {};
		\node [style=none] (25) at (11, 0.5) {};
		\node [style=none] (26) at (4.25, -2) {};
		\node [style=none] (27) at (2.75, -2) {};
		\node [style=square box, minimum width=1 cm, minimum height=0.75 cm] (28) at (8, -0.25) {$g_*$};
		\node [style=square box, minimum width=1 cm, minimum height=0.75 cm] (29) at (10.5, -0.25) {$g$};
		\node [style=white dot] (30) at (3.5, -1) {};
		\node [style=white dot] (31) at (7, 1) {};
		\node [style=white dot] (32) at (11.5, 1) {};
		\node [style=none] (33) at (10, 0.5) {};
		\node [style=none] (34) at (2.75, 2) {};
		\node [style=none] (35) at (8.5, 0.5) {};
		\node [style=none] (36) at (5.5, -0) {$=$};
		\node [style=gray dot] (37) at (3.5, 1) {};
		\node [style=none] (38) at (7.5, 0.5) {};
		\node [style=none] (39) at (10.5, 2.25) {};
		\node [style=square box] (40) at (3.5, -0) {$f$};
		\node [style=none] (41) at (8, 2.25) {};
		\node [style=none] (42) at (10.5, 0.5) {};
		\node [style=none] (43) at (8, 0.5) {};
		\node [style=none] (44) at (6.5, 0.5) {};
		\node [style=none] (45) at (12, 0.5) {};
		\node [style=none] (46) at (6.5, -2) {};
		\node [style=none] (47) at (12, -2) {};
	\end{pgfonlayer}
	\begin{pgfonlayer}{edgelayer}
		\draw [style=diredge] (1) to (2);
		\draw [style=diredge, in=143, out=-90] (4.center) to (2);
		\draw [style=diredge, in=-90, out=37] (2) to (3.center);
		\draw [style=diredge, in=-90, out=150] (5) to (7.center);
		\draw [style=diredge, in=60, out=-90] (6.center) to (5);
		\draw [style=diredge] (9.center) to (11.center);
		\draw [style=diredge] (12.center) to (10.center);
		\draw [style=diredge, in=90, out=90, looseness=1.50] (15.center) to (14.center);
		\draw [style=diredge] (16.center) to (1);
		\draw [in=-90, out=-90, looseness=1.25] (17.center) to (16.center);
		\draw (5) to (17.center);
		\draw (18.center) to (6.center);
		\draw [style=diredge] (19.center) to (20.center);
		\draw [style=diredge] (21.center) to (22.center);
		\draw [style=diredge] (40) to (37);
		\draw [style=diredge, bend right] (34.center) to (37);
		\draw [style=diredge, bend right] (37) to (24.center);
		\draw [style=diredge, bend right] (30) to (27.center);
		\draw [style=diredge, bend right] (26.center) to (30);
		\draw [style=diredge] (30) to (40);
		\draw [style=diredge] (42.center) to (39.center);
		\draw [style=diredge, in=-165, out=90] (25.center) to (32);
		\draw [style=diredge] (41.center) to (43.center);
		\draw [style=diredge, in=90, out=-15] (31) to (38.center);
		\draw [style=diredge, in=90, out=90, looseness=1.50] (33.center) to (35.center);
		\draw [in=90, out=-165] (31) to (44.center);
		\draw [style=diredge, in=-15, out=90] (45.center) to (32);
		\draw [style=diredge] (44.center) to (46.center);
		\draw (47.center) to (45.center);
	\end{pgfonlayer}
\end{tikzpicture}

%% file: figures/stoch_condition.tikz
\begin{tikzpicture}[dotpic]
	\begin{pgfonlayer}{nodelayer}
		\node [style=none] (0) at (0, -0) {$=$};
		\node [style=none] (1) at (-1.25, 2) {};
		\node [style=none] (2) at (-2.75, 2) {};
		\node [style=none] (3) at (-1.25, -2) {};
		\node [style=none] (4) at (-2.75, -2) {};
		\node [style=square box, minimum width=1 cm] (5) at (2, -0) {$g_*$};
		\node [style=none] (6) at (4.5, -0.5) {};
		\node [style=none] (7) at (2, -0.5) {};
		\node [style=none] (8) at (2, -1.75) {};
		\node [style=none] (9) at (4.5, -1.75) {};
		\node [style=none] (10) at (5, 0.5) {};
		\node [style=none] (11) at (2.5, 0.5) {};
		\node [style=none] (12) at (5, 2) {};
		\node [style=none] (13) at (1.5, 2) {};
		\node [style=square box, minimum width=1 cm] (14) at (4.5, -0) {$g$};
		\node [style=none] (15) at (1.5, 0.5) {};
		\node [style=none] (16) at (4, 0.5) {};
		\node [style=square box] (17) at (-2, 0) {$f$};
		\node [style=gray dot] (18) at (3.25, 1.25) {};
		\node [style=white dot] (19) at (-2, -1) {};
		\node [style=black dot] (20) at (-2, 1) {};
	\end{pgfonlayer}
	\begin{pgfonlayer}{edgelayer}
		\draw [style=diredge] (9.center) to (6.center);
		\draw [style=diredge] (10.center) to (12.center);
		\draw [style=diredge] (8.center) to (7.center);
		\draw [style=diredge, in=0, out=90, looseness=1.25] (16.center) to (18);
		\draw [style=diredge, in=180, out=90] (11.center) to (18);
		\draw [style=diredge, in=-135, out=90] (4.center) to (19);
		\draw [style=diredge, in=-45, out=90] (3.center) to (19);
		\draw [style=diredge] (19) to (17);
		\draw [style=diredge] (17) to (20);
		\draw [style=diredge, in=-90, out=45] (20) to (1.center);
		\draw [style=diredge, in=-90, out=135] (20) to (2.center);
		\draw [style=diredge] (15.center) to (13.center);
	\end{pgfonlayer}
\end{tikzpicture}